\newcommand{\cv}[1]{} 
\newcommand{\av}[1]{#1} 

\cv{
\documentclass[twoside,leqno,twocolumn]{article}

\usepackage[letterpaper]{geometry}

\usepackage{ltexpprt}
}

\av{\documentclass[letter]{article}
  \usepackage[totalwidth=420pt,totalheight=625pt]{geometry}
\date{}
}

\usepackage[breaklinks]{hyperref}

\usepackage{xspace}
\usepackage{enumerate}
\usepackage[shortlabels]{enumitem}
\usepackage{multirow}
\usepackage{multicol}
\usepackage{todonotes}
\usepackage[numbers]{natbib}
\usepackage{booktabs}
\usepackage{caption}
\usepackage{subcaption}
\usepackage{lastpage}
\usepackage{thmtools}

\av{\usepackage{amsthm}}

\av{
\newtheorem{theorem}{Theorem}[section]
\newtheorem{lemma}{Lemma}[section]

\newtheorem{Definition}{Definition}[section]
}

\usepackage{tikz}
\usetikzlibrary{decorations.pathmorphing,arrows.meta,fit,calc}
\tikzstyle{node}=[fill=white, draw=black, shape=rectangle, minimum size = 0.4cm]
\tikzstyle{circ}=[fill=white, draw=black, shape=ellipse, minimum size = 0.4cm]
\pgfdeclarelayer{bg}    
\pgfsetlayers{bg,main}  

\usepackage{algorithm}
\usepackage[noend]{algpseudocode}

\usepackage{amsmath}
\usepackage{amssymb}
\usepackage{mathtools}
\usepackage{bbm}
\usepackage{bm}
\usepackage{nicefrac}
\usepackage[noabbrev]{cleveref} 

\newtheorem{example}{Example}
\newtheorem{reduction}{Reduction}

\usepackage{stmaryrd} 
\usepackage{latexsym} 

\newcommand{\DAGger}{\textsc{DAGer}\xspace}
\newcommand{\true}{\textsc{true}}
\newcommand{\false}{\textsc{false}}
\newcommand{\SCC}{\ensuremath{\textsc{SCC}_{||}}}

\newcommand{\SB}{\{\,}%
\newcommand{\SM}{\mid}
\newcommand{\SE}{\,\}}%
\newcommand{\Card}[1]{|#1|}

\usepackage{xcolor}
\colorlet{MyBlue}{blue!50!black!100!}
\colorlet{MyRed}{red!50!black!100!}

\newcommand*{\tikzmk}[1]{\hspace{1em}\tikz[remember picture,overlay,] \node (#1) {};\ignorespaces}
\newcommand{\boxit}[1]{\tikz[remember picture,overlay]{\node[yshift=3pt,fill=#1,opacity=.25,fit={(A)($(B)+(.65\linewidth,-0.2\baselineskip)$)}] {};}\ignorespaces}
\colorlet{MyYellow}{yellow!100}
\colorlet{MyCyan}{cyan!60}

\setlist[itemize]{leftmargin=9pt, noitemsep, topsep=0pt}

\begin{document}

\newcommand\relatedversion{}
\renewcommand\relatedversion{\thanks{The full version of the paper can be accessed at \protect\url{https://arxiv.org/abs/1902.09310}}} 
\cv{
\todo{Add arxiv version}
}

\cv{
\title{A Dynamic MaxSAT-based Approach to Directed Feedback Vertex Sets}
\author{Rafael Kiesel\thanks{Institute of Logic and Computation, TU Wien, Vienna, Austria. The authors acknowledge the support from the Austrian Science Fund (FWF), projects P32441, W1255, and W1255-N23, and from the WWTF, project ICT19-065.} \and Andr\'{e} Schidler\footnotemark[1]}
}

\av{
\title{A Dynamic MaxSAT-based Approach to Directed Feedback Vertex Sets\thanks{The authors acknowledge the support from the Austrian Science Fund (FWF), projects P32441, W1255, and W1255-N23, and from the WWTF, project ICT19-065.}}
\author{%
Rafael Kiesel\\[4pt]
\small Institute of Logic and Computation\\[-3pt]
\small TU Wien, Vienna, Austria\\[-3pt] 
\small \texttt{r.kiesel@tuwien.ac.at}
\and
Andr\'{e} Schidler\\[4pt]
\small Algorithms and Complexity Group\\[-3pt]
\small TU Wien, Vienna, Austria\\[-3pt] 
\small \texttt{aschidler@ac.tuwien.ac.at}
}
}

\date{}

\maketitle

\cv{
\fancyfoot[R]{\scriptsize{Copyright \textcopyright\ 2023 by SIAM\\
Unauthorized reproduction of this article is prohibited}}
}




\begin{abstract}
    We propose a new approach to
    the \emph{Directed Feedback Vertex Set Problem (DFVSP)}, where the 
    input is a directed graph and the solution
    is a minimum set of vertices whose removal makes
    the graph acyclic.
    
Our approach, implemented in the solver \DAGger, is based on two novel contributions:
Firstly, we add a wide range of data reductions that are partially inspired by reductions for the similar vertex cover problem. For this, we give a theoretical 
basis for lifting reductions from vertex cover to DFVSP but also incorporate novel ideas into strictly more general and new DFVSP reductions.

Secondly, we propose dynamically encoding DFVSP in propositional logic using \emph{cycle propagation} for improved performance. Cycle propagation builds on the idea that already a limited
number of the constraints in a propositional encoding is usually sufficient for finding an optimal solution.
Our algorithm, therefore, starts with a small number of constraints
and cycle propagation adds additional constraints when
necessary.
We propose an efficient integration of cycle propagation into the workflow of MaxSAT solvers, further improving the performance of 
our algorithm.
		
	Our extensive experimental evaluation shows that
\DAGger significantly outperforms the state-of-the-art solvers
	and that our data reductions alone directly solve many of the instances.
\end{abstract}


\section{Introduction}
The \emph{Directed Feedback Vertex Set Problem (DFVSP)} is one of Karp's original 21 \textsc{NP}-complete problems~\cite{karp1972reducibility} and has a wide range of applications
such as for argumentation frameworks~\cite{dvorak12, dvorak2022}, deadlock detection, program verification and VLSI chip design~\cite{silberschatz2018operating}.
The problem is to find for a given directed graph a set of vertices---a Directed Feedback Vertex Set (DFVS)---such that every directed cycle uses at least one of the DFVS' vertices. 
In this paper, we consider the optimization version of the problem,
where we search for a minimum DFVS of the graph.

While the problem is known to be fixed parameter tractable~\cite{chen2008fixed} in the
size of the minimum DFVS, these algorithms do not help us in practice~\cite{fleischer2009experimental}, where the size of a minimum DFVS can become prohibitively large. Other possible approaches include explicit branch-and-bound solvers~\cite{lin2000computing}, as well as encodings into a standard format for optimization problems, however, also these do not manage to solve all practically relevant instances~\cite{lin2000computing}.

We attack this problem from two directions.
First, we use knowledge from the \emph{Vertex Cover Problem (VCP)} where the key to success is the application of \emph{data reductions}, which results in
a smaller graph that allows for 
easier solving. The reductions' transfer from VCP to DFVSP is possible, since, VCP is a special case of DFVSP.
And, although many DFVSP instances do not match this special case,
we give a theoretical basis that allows
us to lift many data reductions from VCP to DFVSP.
We use the fact that for these data reductions it is often 
sufficient if the DFVSP instance behaves \emph{locally} like a VCP instance.
We use this to lift several VCP data reductions to DFVS and,
additionally, introduce several more general data reductions.

Our second advancement revisits the idea of using a general purpose optimization solver for DFVSP.
We consider Maximum Satisfiability (MaxSAT) solvers as they
build on top of Propositional Satisfiability (SAT) solvers, which
are designed for efficiently handling binary decisions,
in our case inclusion or exclusion of a vertex, and have
seen tremendous advances in the last decade.

The problem we face, when encoding DFVSP in propositional logic,
is that the resulting formula quickly becomes
prohibitively large.
A straightforward well-performing encoding has size cubic in the number of vertices~\cite{janota2017}.
Worse, the size of the encoding that is most suitable for our purposes can become exponential
in the number of vertices: we list all cycles and state
that any solution must contain at least one vertex from every cycle.

Fortunately, it is known that a subset of an instance's constraints
is often sufficient for finding a feasible solution for the whole
instance.
This is formalized in approaches like \emph{counter-example guided abstraction refinement (CEGAR)}~\cite{clarke2003}.
CEGAR starts with a subset of the constraints, called an under-abstraction,
and whenever the solver returns an infeasible solution, further constraints are added for the next solver run, until
a feasible solution is obtained.
We propose an even faster approach: by implementing \emph{cycle propagation} directly into the MaxSAT solver,
it is possible to immediately add necessary constraints
whenever they are needed, instead of waiting for the solver to finish.
The solver then immediately corrects its decision and
never returns an infeasible solution.

Our implementation won PACE~2022~\cite{pace2022}.

\subsection{Contributions}
Our main contributions are as follows:
\begin{enumerate}[leftmargin=9pt,itemsep=0pt]
    \item We found a general condition that allows us to lift data reductions from the vertex cover problem to the DFVSP problem.  
    \item For many reductions, we went even further and provided non-trivial generalizations of vertex cover reductions that are applicable in a wide range of cases.
    \item We show that implementing cycle propagation is highly beneficial for MaxSAT-based DFVSP-solving, especially since it is tightly integrated into the MaxSAT solver.
    \item In terms of data reductions, our experimental evaluation reveals that our novel reductions can significantly reduce instance sizes, especially on structured instances. 
    \item In terms of cycle propagation, the experiments showed a significant speedup, which is even larger when the cycle propagation happens within the MaxSAT solver.
\end{enumerate}
Overall, the theoretical innovations of this work together with the highly engineered MaxSAT solvers as a basis provide a new and highly efficient strategy of exactly solving DFVSP instances that significantly outperforms any currently available implementation.

\subsection{Related Work}
CEGAR has been successfully used for various other problems such as
Hamiltonian cycles~\cite{soh2014}, graph coloring~\cite{glorian2019}, propositional circumscription~\cite{Janota10},
SMT solving~\cite{brummayer2009},
QBF Solving~\cite{Janota16}, and encoding other PSPACE-hard problems~\cite{Leberre17,Seipp18}.
While CEGAR approaches are sometimes integrated inside the solver (e.g., for SMT solving~\cite{brummayer2009}),
the abstraction refinement is usually added as an extra layer on top of the (Max)SAT solver.
The idea of generating the extra clauses during the solving, before the solver returns the
assignment, is facilitated using \emph{propagators}, which allow for adding extra logic
at different points of the solver's solving process.
This has been successfully used for SMT solving~\cite{brummayer2009}
and dynamic symmetry breaking~\cite{Kirchweger21}.
To our knowledge, cycle propagation (\Cref{sec:maxsat}) is the first integrated 
approach for solving a combinatorial optimization problem.

Regarding DFVSP solving, there has been previous work on data reductions~\cite{koehler2005contraction,lemaic2008markov,levy1988contraction,lin2000computing}, approximate solvers~\cite{even1998approximating,zhou2016spin}, and exact solvers~\cite{bao2018analysis,fages2006constraint,funke1996polyhedral}. Especially for the latter there has been a large increase recently, since DFVSP was the topic of this years PACE Challenge~\cite{pace2022}. Here, many approaches were also based on a CEGAR-like method, i.e., repeatedly
refining a small set of constraints, but using
integer linear programming solvers instead of MaxSAT solvers.
Nevertheless, to the best of our knowledge both our data reductions as well as the idea of using cycle propagation \emph{within} a solver for DFVSP are novel.

\section{Preliminaries}\label{sec:prelims}
\paragraph{(Di)Graphs}
We consider both undirected and directed graphs (digraphs). For a (di)graph $G$, we denote by $V(G)$ its vertices and by $E(G)$ its edges (or arcs if $G$ is a digraph).
Further, we denote an (undirected) edge between vertices $u$ and $v$ as $\{u, v\}$ and the arc from $u$ to $v$ as $(u, v)$. If for an arc $(u,v) \in E(G)$ also the arc in the other direction is present, i.e., $(v,u) \in E(G)$, then we call it a bi-edge, and an arc $(u, u)$ is called a \emph{loop}. The neighbors of a vertex in a graph $G$ are denoted by $N^G(v)$ and for a digraph we use $N^G_{pre}(v), N^G_{succ}(v), N^G_{bi}(v), N^G(v)$ for the set of predecessors, successors, their intersection and their union, respectively. If $G$ is clear from the context, we may omit the superscript.
In a (di)graph a (di)clique is a set of vertices $S$ where each $v \in S$ satisfies $S \setminus \{v\} \subseteq N(v)$.

In the later sections we need some operations to modify (di)graphs. We define 
\begin{itemize}[leftmargin=9pt,itemsep=0pt]
    \item $G - S$ as the (di)graph obtained by removing a set $S$ of vertices or edges from a (di)graph $G$. If $S$ is a singleton set $\{v\}$ or $\{(u,v)\}$ we may omit the brackets and write $G - v$ or $G - (u,v)$.
    \item $G + S$ is defined analogously but adds vertices or edges.
    \item $G[S]$ as the induced sub(di)graph of $G$ with respect to a set of vertices $S$, i.e., $V(G[S]) = S$ and $E(G[S]) = E(G) \cap S\times S$.
    \item $\Pi(G)$ for a digraph $G$ as the graph with $V(\Pi(G)) = V(G)$ and $E(\Pi(G)) = \{ \{u,v\} \SM (u,v), (v,u) \in E(G)\}$.
\end{itemize} 
\bigskip
\paragraph{DFVSP}
We can now introduce the central problem addressed in this paper.
\begin{Definition}[Cycle, DFVS]
Given a digraph $G$ a path is a list of vertices $v_1, \dots, v_n$ such that for $i = 1, \dots, n-1$ there exists an arc $(v_i, v_{i+1}) \in E(G)$. A cycle is a path $v_1, \dots, v_n$ such that $v_1 = v_n$. Furthermore, a path (or cycle) is 
    uncovered, if there is no cycle $v_1', \dots, v_m'$ such that $\SB v_i' \SM i = 1, \dots, m\SE \subsetneq \SB v_i \SM i = 1, \dots, n\SE$
and the \emph{length} of a cycle is the number of distinct vertices in the cycle.
$\mathcal{C}(G)$ refers to the set of all uncovered cycles in $G$. 

A \emph{Directed Feedback Vertex Set (DFVS)} of $G$ is a set $D \subseteq V$ such that every cycle of $G$ contains at least one vertex in $D$. A minimum DFVS is a DFVS of minimum cardinality. We denote by $DFVS(G)$ the minimum cardinality of any DFVS of $G$.
\end{Definition}
An example of a DFVSP-instance with a solution is given in Figures~\ref{fig:example:graph} and~\ref{fig:example:dag}.
As an example for covered and uncovered cycles, consider the cycle $a,c,b,a$.
The cycle is covered, as vertices of both (uncovered) cycles $a,b,a$ and $c,b,c$ are proper
subsets of $\SB a, b, c \SE$.

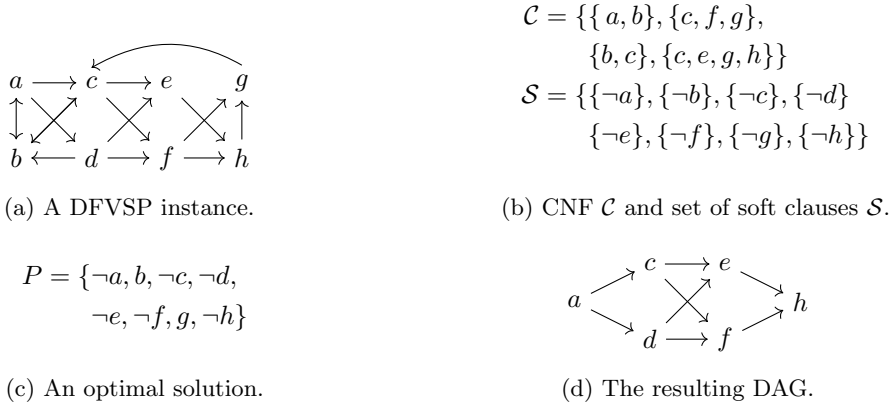
\begin{figure}
\begin{subfigure}[b]{0.39\linewidth}
    \centering
        \begin{tikzpicture}
        \node (a) at (0, 0) {$a$};
        \node (b) at (0, -1) {$b$};
        \node (c) at (1, 0) {$c$};
        \node (d) at (1, -1) {$d$};
        \node (e) at (2, 0) {$e$};
        \node (f) at (2, -1) {$f$};
        \node (g) at (3, 0) {$g$};
        \node (h) at (3, -1) {$h$};
        \draw[<->] (a) -- (b);
        \draw[->] (a) -- (c);
        \draw[->] (c) -- (b);
        \draw[->] (b) -- (c);
        \draw[->] (a) -- (d);
        \draw[->] (c) -- (e);
        \draw[->] (c) -- (f);
        \draw[->] (d) -- (e);
        \draw[->] (d) -- (f);
        \draw[->] (d) -- (b);
        \draw[->] (e) -- (h);
        \draw[->] (f) -- (h);
        \draw[->] (f) -- (g);
        \draw[->] (h) -- (g);
        \draw[->] (g.north) to [out=150,in=30]  (c.north);
    \end{tikzpicture}
    \caption{A DFVSP instance.}
    \label{fig:example:graph}
\end{subfigure}
\begin{subfigure}[b]{0.59\linewidth}
\begin{align*}
    \mathcal{C} =  \{ &\SB a, b\},  \{ c, f, g \}, \\
    & \{ b, c \}, \{ c, e, g, h \} \} \\
    \mathcal{S} = \{ &\{\lnot a\}, \{\lnot b\}, \{\lnot c\}, \{\lnot d\}  \\
    &\{\lnot e\}, \{\lnot f\}, \{\lnot g\}, \{\lnot h\} \}
\end{align*}
\caption{CNF $\mathcal{C}$ and set of soft clauses $\mathcal{S}$.}
\label{fig:example:encoding}
\end{subfigure}
\begin{subfigure}[b]{0.4\linewidth}
\begin{align*}
P = \{ &\lnot a, b, \lnot c , \lnot d, \\ & \lnot e, \lnot f, g, \lnot h \}
\end{align*}
\caption{An optimal solution.}
\label{fig:example:model}
\end{subfigure}
\begin{subfigure}[b]{0.58\linewidth}
\centering
        \begin{tikzpicture}
        \node (a) at (0, -0.5) {$a$};
        \node (c) at (1, 0) {$c$};
        \node (d) at (1, -1) {$d$};
        \node (e) at (2, 0) {$e$};
        \node (f) at (2, -1) {$f$};
        \node (h) at (3, -0.5) {$h$};
        \draw[->] (a) -- (c);
        \draw[->] (a) -- (d);
        \draw[->] (c) -- (e);
        \draw[->] (c) -- (f);
        \draw[->] (d) -- (e);
        \draw[->] (d) -- (f);
        \draw[->] (e) -- (h);
        \draw[->] (f) -- (h);
    \end{tikzpicture}
    \caption{The resulting DAG.}
    \label{fig:example:dag}
\end{subfigure}
\caption{Example of a DFVSP instance, a solution, and its encoding as a MaxSAT instance.}
    \label{fig:example}
\end{figure}

We introduce problems (Vertex Cover, Propositional Satisfiability)
and relate them to DFVSP.
Throughout this paper,
we make this correspondence clearer
by using similar names for corresponding objects,
e.g., $C_i$ for a clauses in propositional logic as they
correspond to cycles.

Whenever each uncovered cycle has length 2 and using $\Pi(G)$,
we can state DFVSP as follows:
\begin{Definition}[Vertex Cover]
Let $G$ be an undirected graph. A minimum Vertex Cover (minimum VC) is a set $D \subseteq V(G)$, such
that $D$ is of minimum cardinality and for each edge $\{u, v\} \in E(G)$ it holds that $\{u,v\} \cap D \neq \emptyset$. $VC(H)$ denotes the minimum cardinality over all VCs of $H$.
\end{Definition}

\paragraph{Propositional Logic and SAT solvers}
We give a brief introduction to propositional logic, SAT solvers, and MaxSAT solvers and refer the interested reader to~\cite{sathandbook2021} for more details.

We use propositional formulas in Conjunctive Normal Form (CNF). A CNF $\mathcal{C}$, defined for a set $V$ of variables, is a finite conjunction of \emph{clauses} $C_i$, where each clause consists of a finite disjunction of \emph{literals} $\ell \in \{v, \lnot v\}$ for some $v \in V$. 
We use the equivalency $\lnot\lnot v = v$ and say a literal $\ell$ is \emph{negative} if it is of the form $\ell = \lnot v$ and \emph{positive}
if $\ell = v$.
Further, we denote clauses as a set of literals.

We represent truth assignments as a subset of $V \cup \SB \lnot v \SM v \in V \SE$ that for any $v \in V$ do not contain both $v$ and $\lnot v$.
Here, a positive literal indicates the value true, a negative literal the value false,
and if a variable does not occur in the assignment, it isn't assigned a value (yet).
We use the standard satisfaction relation and call an assignment that satisfies a formula a \emph{model}.

A \emph{propositional satisfiability (SAT) solver} takes as input a CNF and
checks whether it has a model and returns one, if so.
\emph{MaxSAT} solvers additionally search for a model that optimizes an objective function.
We use \emph{partial MaxSAT}, where
the solver takes two CNFs, called the \emph{hard} and \emph{soft} clauses.
Then, a MaxSAT model is a model of the hard clauses that satisfies as many soft clauses as possible.

Consequently, we encode DFVSP as a MaxSAT instance by using $V(G)$
as the variables and adding for each uncovered cycle $v_1\dots v_n \in \mathcal{C}(G)$ a corresponding hard clause $\SB v_1, \dots, v_n\SE$. Thus, if $v_i$ is true it is in the DFVS. Accordingly, to achieve minimum cardinality we add a soft clause $\{\lnot v\}$ for each $v \in V(G)$.
A MaxSAT model then maximizes the elements that are not in the DFVS,
effectively minimizing the elements that are in the solution.
In Figures~\ref{fig:example:encoding} and~\ref{fig:example:model}
we see an example of a MaxSAT encoding and the corresponding model.

\section{Solver Architecture and Outline}
We give a brief overview of our algorithm
which also determines the structure of the paper.
The general algorithm is shown in \Cref{alg:main-algo}.

\begin{algorithm}[tbh!]
\begin{algorithmic}[1]
    \State $G \gets $ \Call{Reduce}{$G$}
    \State $\mathcal{C}' \gets$ \Call{FindShortCycles}{$G$}
    \If{$\mathcal{C}' = \mathcal{C}(G)$}
        \State $G, \mathcal{C}' \gets $ \Call{ReduceWithAllCycles}{$G$, $\mathcal{C}'$}
    \EndIf
    \State $D \gets $\Call{DFVS\_MaxSAT}{$G, \mathcal{C}'$}
    \State \Return $D$
\end{algorithmic}
\caption{The complete algorithm.}
\label{alg:main-algo}
\end{algorithm}

The algorithm starts with performing general data reductions,
discussed in \Cref{sec:reductions} and then searches
for a small set of short cycles.

Short cycles are cycles with a specified maximum length.
Bounded depth first search can easily find these short cycles,
as long as the bound, i.e., cycle length, is small enough.
We discuss the associated limits on both the cycle length and
the number of cycles in the implementation details in \Cref{sec:experiments}.

Should the set of short cycles represent all uncovered cycles,
we can apply additional data reductions, which are also discussed in \Cref{sec:reductions}.
We increase our chances of finding all uncovered cycles by ignoring the
aforementioned limits and incrementally extending the maximum length of
a cycle, as long as the number of new cycles we find decreases monotonically.

A MaxSAT solver is then used to solve the reduced instance.
The set $\mathcal{C}' \subseteq \mathcal{C}(G)$ is
given to the solver as the initial set of constraints.
In case $\mathcal{C}'$ contains all uncovered cycles,
the MaxSAT solver simply computes the minimum DFVS.
When $\mathcal{C}'$ is not complete, the cycle propagation
discussed in \Cref{sec:maxsat} ensures that
the solver still returns a valid DFVS.

We evaluate how effective our solver is in \Cref{sec:experiments}.

\section[Data Reductions]{Data Reductions%
\av{\footnote{Proofs for all theorems are in \Cref{sec:proofs}}}%
\cv{\footnote{Proofs for all theorems are in the long-version.}}
}
\label{sec:reductions}
Data reductions are the first step in our approach.
They aim to shrink the input digraph in such a way that 
a minimum DFVS for the reduced digraph can easily be extended 
to a minimum DFVS for the original digraph.
To this end, we apply a wide range of reduction rules. 
We give a list of all relevant reductions in \Cref{tab:reduction_rules}.

\begin{table*}[tbh]
    \centering
    \begin{tabular}{c c c}
        \toprule
        Name & Origin & Strictly Subsumed by \\
        \midrule
        LOOP & \cite{levy1988contraction} & - \\
        IN0/1 & \cite{levy1988contraction} & INDICLIQUE \\
        OUT0/1 & \cite{levy1988contraction} & OUTDICLIQUE \\
        INDICLIQUE & \cite{lemaic2008markov} & - \\
        OUTDICLIQUE & \cite{lemaic2008markov} & - \\
        DICLIQUE-2 & \cite{lemaic2008markov} & - \\
        DICLIQUE-3 & \cite{lemaic2008markov} & - \\
        PIE & \cite{lin2000computing} & ALLCYCLES \\
        DOME & \cite{lin2000computing} & ALLCYCLES \\
        DOME++ & - & ALLCYCLES \\
        ALLCYCLES & - & - \\
        CORE & \cite{lin2000computing} & IN/OUTDICLIQUE\\
        ``Reduction 2'' & \cite{stege1999improved} & SUBSET (DFVSP) \\
        SUBSET (VCP) & \cite{stege1999improved} & SUBSET (DFVSP)\\
        SUBSET (DFVSP) & - & -\\
        2FOLD & \cite{xiao2013confining} & MANYFOLD (DFVSP) \\
        ``Reduction 4'' & \cite{stege1999improved} & MANYFOLD (DFVSP) \\
        ``Reduction 5'' & \cite{stege1999improved} & MANYFOLD (DFVSP) \\
        ``Reduction 7.2'' & \cite{stege1999improved} & MANYFOLD (DFVSP) \\
        MANYFOLD (VCP) & \cite{fellows2018what} & MANYFOLD (DFVSP) \\
        MANYFOLD (DFVSP) & - & - \\
        4PATH (VCP) & \cite{fellows2018what} & 4PATH (DFVSP) \\
        4PATH (DFVSP) & - & - \\
        UNCONFINED (VCP) & \cite{xiao2013confining} & UNCONFINED (DFVSP) \\
        UNCONFINED (DFVSP) & - & - \\
        3EMPTY & \cite{stege1999improved} & - \\
        TWIN & \cite{xiao2013confining} & 3EMPTY + MANYFOLD (DFVSP) \\
        FUNNEL & \cite{xiao2013confining} & SUBSET + MANYFOLD (DFVSP) \\
        DESK & \cite{xiao2013confining} & - \\
        \bottomrule
    \end{tabular}
    \caption{A summary of all used (and some unused) reduction rules, their origin, and the (currently) most general rule subsuming it.}
    \label{tab:reduction_rules}
\end{table*}

\subsection{DFVSP Reductions}
For DFVSP there is already a wide range of rules by~\citet{levy1988contraction,lemaic2008markov,lin2000computing}, which we use in an unmodified manner. 
\av{For space reason, we do not repeat them here but refer to \Cref{app:reductions} for details.}
\cv{For space reasons, we do not repeat them here.}


Apart from this, we generalized DOME~\cite{lin2000computing} from arcs dominated by length 2 paths to arcs dominated by arbitrary length paths:
\begin{reduction}[DOME++] 
If there is an arc $(v,u) \in E(G)$ such that $(u,v) \not\in E(G)$ and (i) every path that starts at $v$ and ends at $u$ uses a bi-edge or a vertex from $N_{pre}(u) \setminus \{v\}$ or (ii) every path that starts at $u$ and ends at $v$ uses a bi-edge or a vertex from $N_{succ}(v) \setminus \{u\}$, then replace $G$ by $G - (v,u)$.
\end{reduction}
Additionally, while enumerating all uncovered cycles is generally not feasible, due to their 
potentially exponential number, it is often possible in practice. This allows the following reduction.
\begin{reduction}[ALLCYCLES]
If there is an arc $(v,u) \in E(G)$ such that every cycle that visits $v$ immediately after $u$ is covered, then replace $G$ by $G - (v,u)$.
\end{reduction}
\begin{theorem}
DOME++ and ALLCYCLES are sound if $G$ is loop-free.
\end{theorem}
Apart from the DFVSP reductions, we also lifted and generalized VCP reductions.

\subsection{VCP Reductions}
As noted already in the preliminaries, VCP and DFVSP are related. We formalized this intuition as follows:
\begin{lemma}
\label{lem:split}
Let $G$ be a digraph, then $S \subseteq V(G)$ is a DFVS if and only if 
 \begin{itemize}
    \item $S$ is a VC of $\Pi(G)$, and
    \item $S$ is a DFVS of $G - E(\Pi(G))$.
\end{itemize}
\end{lemma}
Thus, we can treat a DFVSP instance as a combination
of a VCP instance $\Pi(G)$ and a smaller DFVSP instance without bi-edges.
Hence, given sufficient preconditions, it suggests itself to apply
VCP reductions to DFVSP.
We derive these preconditions using boundary reductions similar to those of~\cite{fellows2018what}, which intuitively locally replace one part of a graph by another.
\begin{Definition}[Boundary VCP Reduction]
A boundary VCP reduction is a tuple $r = \langle H, H', B, c\rangle$, where $H, H'$ are graphs, $B \subseteq V(H)\cap V(H')$ is a set of non-isolated vertices in both $H$ and $H'$, and $c \in \mathbb{Z}$, such that 
for all $X \subseteq B$ it holds that 

\smallskip\noindent
\centerline{$
VC(H - X) = VC(H' - X) + c.
$}

\smallskip\noindent
A reduction $r$ is \emph{applicable} to $G$, if the vertices in $B$ are (i) not isolated in $G$, (ii) an independent set in $G$, and (iii) the overlapping vertices, i.e., $V(G) \cap V(H) = B = V(G) \cap V(H')$.
\end{Definition}
\begin{example}
An example of a boundary reduction $r = \langle H, H', B, c\rangle$ is given by 
the boundary $B = \{b_1,b_2\}$, size difference $c = 1$, and the graphs
\begin{center}
\begin{tikzpicture}
		\node (h) at (-1, 0) {$H = $};
        \node (b1) at (0, 0) {$b_1$};
        \node (i1) at (1, 0) {$v_1$};
        \node (i2) at (2, 0) {$v_2$};
        \node (b2) at (3, 0) {$b_2$};
        \draw[-] (b1) -- (i1);
        \draw[-] (i1) -- (i2);
        \draw[-] (i2) -- (b2);
        \node (hp) at (4, 0) {$H' = $};
        \node (bp1) at (5, 0) {$b_1$};
        \node (bp2) at (6, 0) {$b_2$};
        \draw[-] (bp1) -- (bp2);
\end{tikzpicture}
\end{center}
$r$ is a boundary reduction, since for $X \subseteq \{b_1,b_2\}, X \neq \emptyset$, there is always still the edge between $v_1$ and $v_2$ in $H-X$, which means that $VC(H-X) = 1$, whereas $H'-X$ is edgeless. On the other hand, if $X = \emptyset$, then $VC(H-X) = VC(H) = 2$ and $VC(H'-X) = VC(H') = 1$.
\end{example}
Importantly, applicability guarantees soundness and the desired locality property:
\begin{theorem}
\label{thm:vc_boundary}
For every graph $G$ such that $\langle H, H', B, c\rangle$ is applicable
it holds that for every minimum VC $S$ of $G + V(H) + E(H)$ there is a minimum VC $S'$ of $G + V(H') + E(H')$ such that $|S| = |S'| + c$, and $S \cap V(G) = S' \cap V(G)$.
\end{theorem}
Not only does this theorem guarantee that the size of a minimum VC changes by $c$, it also tells us that the modification only has local effects on the minimum VCs. 
This locality is particularly interesting for lifting VCP reductions
to DFVSP, as it allows their application when a digraph ``locally behaves like a VCP instance''. This intuition is formalized in the following theorem.
\begin{theorem}
\label{thm:dfvs_boundary}
Let $G$ be a digraph, $r = \langle H, H', B, c \rangle$ be a boundary VCP reduction and $\Pi(G) = G' + V(H) + E(H)$
such that $r$ is applicable to $G'$.
If\\
(i) all edges incident in $G$ to any $v \in V(H) \setminus B$ are bi-edges and \\
(ii) for every arc $(u,w) \in E(G)$ at least one of the following holds: \\
(ii.a) $(u,w)$ is a bi-edge or \\
(ii.b) $|\{u,w\} \cap B| \leq 1$
then 

\smallskip
\centerline{
$
DFVS(G) = DFVS(G^*) + c,
$
}

\smallskip\noindent
where $G^*$ is given by

\smallskip
\centerline{
$
\begin{array}{rl}
G & - (V(H) \setminus B) - B \times B\\
  & + V(H') + \SB (u,v) \SM \{u,v\} \in E(H')\SE.
\end{array}
$
}
\end{theorem}
We can, thus, use many VCP reductions without modification by checking the preconditions of Theorem~\ref{thm:dfvs_boundary}.
For space reasons, we do not introduce VCP reductions that we strictly generalize to the DFVSP setting later on. \av{They can, however, be found in Appendix~\ref{app:vc-reductions}.}
This leaves only the following reduction, which we apply without
any changes:
\begin{reduction}[3EMPTY~\cite{stege1999improved,fellows2018what}]
If there exists a vertex $v \in V(G)$ such that $|N(v)| = 3$ and $|E(G[N(v)])| = 0$, then replace $G$ by 

\smallskip\noindent
\[
G - v + \{\{a,b\}, \{b,c\}\} + \{a\}\times N(b) + \{b\}\times N(c) + \{c\}\times N(a). 
\]
\end{reduction}
\subsection{Directed Versions of VCP Reductions}
Some VCP reductions can be generalized to DFVSP, even when Theorem~\ref{thm:dfvs_boundary} does not apply. Note that all of the following reductions are strict generalizations of VCP reductions. I.e., when a digraph only has bi-edges, then each of the new DFVS reductions corresponds to a VCP reduction.

A simple example of this is the SUBSET reduction:
\begin{reduction}[SUBSET]
If there exists $v,u \in V(G)$ such that $(v,u), (u,v) \in E(G)$, $N_{pre}(v) \subseteq N_{pre}(u) \cup \{u\}$, and $N_{succ}(v) \subseteq N_{succ}(u) \cup \{u\}$, then replace $G$ by $G-u$.
\end{reduction}
\begin{theorem}
Let $G$ be a digraph. After applying SUBSET to vertices $v,u$ resulting in $G'$, it holds that for every minimum DFVS $S$ of $G'$ the set $S\cup \{u\}$ is a minimum DFVS of $G$.
\end{theorem}
Other reductions such as the MANYFOLD reduction, are more advanced.
\begin{reduction}[MANYFOLD]
If there exists a vertex $v \in V(G)$ such that $N(v) = N_{bi}(v)$ and 
there is a partition
$(C_1,C_2)$ of $N(v)$, where
\begin{itemize}[leftmargin=9pt,itemsep=0pt]
    \item $|C_1| \geq |C_2|$,
    \item $G[C_i]$ is a diclique for $i = 1,2$,
    \item $M$ is the set of non-arcs of $G[N(v)]$,
    \item $(c,d) \in M$ implies that either $(d,c) \in M$ and there is no uncovered path between $c$ and $d$, or $(d,c) \not\in M$ and every uncovered path from $d$ to $c$ uses the arc $(d,c)$
    \item for each $c_1 \in C_1$, there is exactly one $c_2 \in C_2$ (denoted $c_2(c_1)$) such that $(c_1,c_2) \in M$ or $(c_2,c_1) \in M$,
\end{itemize}
then, replace $G$ by 
\begin{align*}
& G - v - C_2 \\
+ & \textstyle\bigcup_{c_1 \in C_1} \{c_1\} \times N_{succ}(c_2(c_1)) \cup  N_{pre}(c_2(c_1)) \times \{c_1\} \\
- & \textstyle\bigcup_{c_1 \in C_1} (c_1,c_1).
\end{align*}
\end{reduction}
We go over the conditions to explain their relevance.
\begin{itemize}
    \item $N(v) = N_{bi}(v)$ needs to hold, to ensure that when a minimum DFVS $S$ does not contain $v$, then it must be the case that $N(v)$ is a subset of $S$.
    \item For each $c_1 \in C_1$ there is exactly one $c_2 \in C_2$ with missing arc $(c_1,c_2)$ or $(c_2,c_1)$. This ensures that when $c_1$ is not in a minimum DFVS $S$, then either all vertices from $C_2$ are in $S$ or only the uniquely determined vertex $c_2$ is not in $S$.
    \item The conditions on $M$ ensure that when we perform the contraction, we only add cycles for which there exists a corresponding cycle in the original digraph.
\end{itemize}
Note that the conditions on the arcs in $M$ are NP-hard to check. Later, we discuss alternative tractable and sufficient conditions. First we state soundness.
\begin{theorem}
Let $G$ be a loop-free digraph, such that PIE is not applicable and MANYFOLD is applicable to $v^{*} \in V(G)$ and $G'$ be the graph obtained from $G$ after MANYFOLD was applied on vertex $v^*$, then $DFVS(G) = DFVS(G') + |C_2|$ and given a minimum DFVS of $G'$, we can in polynomial time compute a minimum DFVS of $G$.
\end{theorem}
We need to check two possible conditions on the arcs in $M$. First, if $(d,c) \in E(G)$ we use straightness:
\begin{Definition}[Straightness]
Let $G$ be a digraph and $(d,c) \in E(G)$. Then $(d,c)$ is straight, if $(c,d) \not \in E(G)$ and (i) every arc $(d,c') \in E(G)$ such that $c'\neq c$ is a bi-edge or (ii) every arc $(d',c) \in E(G)$ such that $d'\neq d$ is a bi-edge.
\end{Definition}
As desired, if the arc $(d,c)$ is straight, then every uncovered path that contains $d$ after $c$ uses it.
If, on the other hand, $(d,c)$ is not in $E(G)$, we need to prohibit the existence of an uncovered path between $c$ and $d$. 
Here, we give a sufficient condition using \emph{Strongly Connected Components (SCCs)}. Recall that an SCC is a subset maximal set $S$ of vertices such that for every combination $v,u$ of vertices in $S$ there is a (directed) path from $u$ to $v$.



We consider for a digraph $G$ the SCCs of $G - E(\Pi(G))$, i.e., $G$ without bi-edges, and denote for a vertex $v \in V(G)$ by $\SCC^{G}(v)$ the unique SCC containing it. If $G$ is clear from the context, we may omit the superscript. 
Then, the PIE reduction~\cite{lin2000computing} allows us to remove arcs between vertices $u,v$ such that $\SCC(v) \neq \SCC(u)$. This entails that when $\SCC(v) \neq \SCC(u)$, every path between $u$ and $v$ uses a bi-edge and is thus covered.

Together these conditions give us a tractable way of guaranteeing applicability of MANYFOLD regardless of whether both $(c,d)$ and $(d,c)$ are in $M$ or whether only one of them is. Figure~\ref{fig:2folds} shows example applications.
\begin{figure}
\begin{subfigure}[b]{\linewidth}
    \centering
    \begin{tikzpicture}
        \node (v) at (0, 0) {$v$};
        \node (a1) at (-1, 1) {$a_1$};
        \node (a2) at (-2, 1) {$a_2$};
        \node (a3) at (-2, 0) {$a_3$};
        \node (b1) at (1, 1) {$b_1$};
        \node (b2) at (2, 1) {$b_2$};
        \node (b3) at (2, 0) {$b_3$};
        \draw[<->] (v) -- (a1);
        \draw[<->] (v) -- (b1);
        \draw[->] (a1) -- (a2);
        \draw[->] (a2) -- (a3);
        \draw[->] (a3) -- (a1);
        \draw[->] (b1) -- (b2);
        \draw[->] (b2) -- (b3);
        \draw[->] (b3) -- (b1);
        \node (a1p) at (4, 1) {$a_1$};
        \node (a2p) at (3, 1) {$a_2$};
        \node (a3p) at (3, 0) {$a_3$};
        \node (b2p) at (5, 1) {$b_2$};
        \node (b3p) at (5, 0) {$b_3$};
        \draw[->] (a1p) -- (a2p);
        \draw[->] (a2p) -- (a3p);
        \draw[->] (a3p) -- (a1p);
        \draw[->] (a1p) -- (b2p);
        \draw[->] (b2p) -- (b3p);
        \draw[->] (b3p) -- (a1p);
    \end{tikzpicture}
    \caption{A MANYFOLD reduction, where $\SCC(a_1) = \{a_1,a_2,a_3\} \neq \{b_1,b_2,b_3\} = \SCC(b_1)$.}
    \label{fig:example:manyfold_I}
\end{subfigure}
\begin{subfigure}[b]{\linewidth}
\centering
\begin{tikzpicture}
        \node (v) at (0, 0.5) {$v$};
        \node (a) at (-1, 1) {$a$};
        \node (b) at (1, 1) {$b$};
        \node (c) at (0, 1.5) {$c$};
        \draw[<->] (v.west) to [out=180,in=270] (a.south);
        \draw[<->] (v.east) to [out=0,in=270] (b.south);
        \draw[->] (a) to (b);
        \draw[->] (b.north) to [out=120,in=0] (c.east);
        \draw[->] (c.west) to [out=180,in=60] (a.north);
        \node (ap) at (4, 0.5) {$a$};
        \node (cp) at (4, 1.5) {$c$};
        \draw[<->] (ap) -- (cp);
    \end{tikzpicture}
\caption{A MANYFOLD reduction, where the arc $(a,b)$ is straight.}
\label{fig:example:manyfold_II}
\end{subfigure}    
\caption{Two applications of the MANYFOLD reduction. In both cases, left is before, right is after.} 
    \label{fig:2folds}
\end{figure}
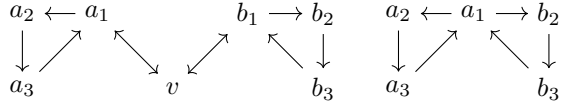
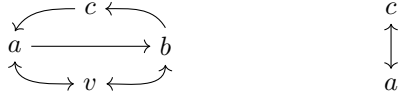

We note that it is not necessary to recompute SCCs at every step, since we can update them after each reduction in an approximate but safe manner and only recompute them periodically.

Apart from MANYFOLD, we can also generalize 4PATH in a similar manner, by exploiting the lack of uncovered paths between some of the involved vertices.
\begin{reduction}[4PATH]
If there exists a vertex $v \in V(G)$ such that
 \begin{itemize}
    \item $N(v) = N_{bi}(v) = \{a,b,c,d\}$,
    \item $E(G[N(v)]) = \SB(a,b),(b,a),(b,c),(c,b),(c,d),(d,c)\SE$,
    \item and there is no uncovered path between any pair of vertices from $\SB\{a,c\}, \{a,d\}, \{d,b\} \SE$,
\end{itemize}
then replace $G$ by 

\smallskip
\centerline{
$
\begin{array}{rl}
    G \;- &\!\!\!\! v + \{(a,c), (c,a), (a,d), (d,a), (b,d), (d,b)\} \\
    + &\!\!\!\!\{a,b\}\times N_{succ}(d) + N_{pred}(d)\times \{a,b\}\\
    + &\!\!\!\!\{c,d\}\times N_{succ}(a) + N_{pred}(a)\times \{c,d\}. 
\end{array}
$
}
\end{reduction}
Again, we practically ensure that every path is covered by requiring $\SCC(a) \neq \SCC(c), \SCC(a) \neq \SCC(d)$ and $\SCC(d) \neq \SCC(b)$.
For either condition 4PATH is sound.
\begin{theorem}
Let $G$ be a digraph such that 4PATH is applicable to $v^* \in V(G)$ and $G'$ be the graph obtained from $G$ after 4PATH was applied to the vertex $v^*$, then
 \begin{itemize}
    \item $DFVS(G) = DFVS(G')$, and
    \item given a minimum DFVS of $G'$, we can in polynomial time compute a minimum DFVS of $G$.
\end{itemize}
\end{theorem}
\begin{algorithm}[t]
\begin{algorithmic}[1]
\algnewcommand{\algorithmicgoto}{\textbf{go to}}%
\algnewcommand{\Goto}[1]{\State \algorithmicgoto~#1}%

\Function{CheckUnconfined}{$v, G$}
\State $A \leftarrow \{v\}$
\State $N \leftarrow \SB u \in V(G) \SM G[A\cup \{u\}] \text{ is cyclic}\SE$
\State $P \leftarrow  \{u \in N |\; |N_{succ}(u) \cap A| + |N_{pred}(u) \cap A| = 2\}$
\If{$P \neq \emptyset$}
    \State $u \leftarrow  \text{argmin}_{u' \in P} |N(u') \setminus (N \cup A)|$
    \If{$|N(u) \setminus (N \cup A)| = 0$}
        \Return True
    \ElsIf{$|N(u) \setminus (N \cup A)| = 1$}
        \State $A \leftarrow A \cup \{u\}$
        \Goto 3
    \EndIf
\EndIf
\Return False
\EndFunction
\end{algorithmic}
\caption{Checks if $v$ is unconfined in a digraph $G$.}
\label{alg:dunconfined}
\end{algorithm}
Whereas the above reductions all capture fixed graph patterns, the applicability of the following one is determined by the iterative procedure in~\Cref{alg:dunconfined}.
\begin{reduction}[UNCONFINED]
If there is a vertex $v \in V(G)$ such that \textsc{CheckUnconfined}($v,G$) returns True, replace $G$ by $G - v$.
\end{reduction}
When a vertex $v$ is unconfined, this guarantees us that while there may be minimum DFVSs that do not contain $v$, there is at least one, which does.
\begin{theorem}
Let $G$ be a digraph. After applying UNCONFINED to vertex $v$ resulting in $G'$, it holds that for every minimum DFVS $S$ of $G'$ the set $S\cup \{v\}$ is a minimum DFVS of $G$.
\end{theorem}
This concludes the data reductions that we use and brings us to the solving step.

\section{MaxSAT Solver}\label{sec:maxsat}
We compute the minimum DFVS for the reduced instance using a MaxSAT solver.
Recall from \Cref{sec:prelims} that we add one disjunction per cycle
containing exactly the variables corresponding to the vertices in the cycle.
Since there is a 1:1 correspondence between vertices and variables, cycles and clauses, 
as well as model, and DFVS, we treat them synonymously in this section.
We also refer to a DFVS candidate that does not break all cycles as \emph{infeasible}
and to a DFVS as a \emph{feasible} solution.

Enumerating all cycles for a complete encoding is generally impossible in practice.
A well-established technique that deals with this issue is CEGAR~\cite{clarke2003}.
In the context of DFVSP, a CEGAR approach initially gives the solver a small,
usually not comprehensive, set of uncovered cycles.
Whenever the solver returns a solution that is not a valid DFVS,
we add cycles that are are not broken by the infeasible solution.
This is repeated until the solver returns a feasible solution.
Very often, a comparatively small number of constraints, in our case cycles,
is sufficient for finding a feasible solution.

The main drawback of this CEGAR approach is the computational overhead.
While a solver's decision may quickly imply that the solution is infeasible,
the solver may run long past that point until it returns a solution.
Further, after an infeasible solution is returned, it is hard to determine
which of the solver's decisions caused the infeasability.
Lacking this knowledge, we have to add many cycles that are not necessary for guiding
the solver towards a feasible solution.

We propose cycle propagation for improved performance.
Cycle propagation adds the feasibility check
directly into the MaxSAT solver's logic and adds the necessary cycles
at exactly the point where the MaxSAT solver's decision would
cause the solution to become infeasible.

We focus on \emph{core-guided} MaxSAT solvers.
Here, the MaxSAT solver implements the search for an optimal solution
and calls the SAT solver repeatedly.
For each SAT call, the MaxSAT solver extends the input CNF by extra
clauses related to the search for an optimal solution~\cite{Morgado14b}.
We therefore add cycle propagation to the SAT solver, as 
the decisions that lead to infeasability are made here.
In order to introduce cycle propagation, we first discuss the basics of CDCL, 
the algorithm used by most modern SAT solvers~\cite{Moskewicz01,Marques-SilvaS99}.

\subsection{Conflict Driven Clause Learning (CDCL)}
\begin{algorithm}[t]
\begin{algorithmic}[1]
\Function{DFVS\_CDCL}{$G, \mathcal{C}$}
\State $D \gets \emptyset$
\While{$|\mathcal{C}| > 0$}
	    \State $\mathcal{C} \gets \mathcal{C} \setminus \SB C \in \mathcal C \SM D \cap C \neq \emptyset \SE$
	    \State $D' \gets \SB \lnot \ell \SM \ell \in D \SE$\label{alg:cdcl:removesat}
	    \If{$\Card{\SB C \in \mathcal{C} \SM \Card{C \setminus D'} = 0 \SE} > 0$}\label{alg:cdcl:conflict}
	        \If{No Decisions} \label{alg:cdcl:unsat}
                    \Return \false
                \EndIf
                \State $\mathcal{C} \gets \mathcal{C} \cup \Call{analyzeConflict}{ }$
                \State $\mathcal{C}, D \gets \Call{backtrack}{ }$
	    \ElsIf{$\Card{\SB C \in \mathcal{C} \SM \Card{C \setminus D'} = 1 \SE} > 0$}\label{alg:cdcl:bcp}
	        \State $D \gets \Call{booleanClausePropagation}{ }$
	    \Else
	        \Statex\tikzmk{A}\hspace{3.6em}\(\triangleright\) Cycle Propagation
    	    \State $V' = \SB v \in V(G) \SM \lnot v \in D \SE$\label{alg:cdcl:dfvs}
    	    \If{$G[V']$ contains a cycle $C$}
    	        \State $\mathcal{C} \gets \mathcal{C} \cup C$
    	    \Else\tikzmk{B}\boxit{MyYellow}  \label{alg:cdc:dfvs-end}
                \State $D \gets D \cup \{\Call{decideLiteral}{ }\}$\label{alg:cdcl:decide}
    	    \EndIf
	    \EndIf
\EndWhile
\State \Return \true
\EndFunction
\end{algorithmic}
\caption{The modified CDCL algorithm.}
\label{alg:cdcl}
\end{algorithm}

We limit ourselves to a cursory discussion of CDCL that
introduces the necessary concepts to understand cycle propagation.
Remember from \Cref{sec:prelims} that a SAT instance consists of variables $V$ and clauses $\mathcal{C}$.
The whole algorithm including cycle propagation is shown in \Cref{alg:cdcl}.
Ignoring cycle propagation in the block starting at Line~\ref{alg:cdcl:dfvs},
the listing shows the basic CDCL algorithm.

CDCL incrementally extends a partial assignment $D$, assigning values to some of the variables,
to a full assignment until it either obtains a model or knows that the formula is unsatisfiable.
The algorithm only keeps unsatisfied clauses and removes satisfied ones (Line~\ref{alg:cdcl:removesat}).

Conflicts occur when a clause $C$ cannot be satisfied by any extension of $D$ to a full assignment, because
$D$ contains the negation of $C$'s literals, as is checked in Line~\ref{alg:cdcl:conflict}.
Here, two things can happen.
If the conflict occurred without any prior decision, the set of clauses
implies a conflict and the formula is unsatisfiable.
Otherwise, CDCL learns a conflict clause: a clause based on the decisions
that lead to the conflict and that prevents the solver from making the same set of
decisions again.
Afterwards, the solver backtracks, where it removes the corresponding literals from $D$ and restores the corresponding removed clauses and literals to $\mathcal{C}$.

\emph{Boolean constraint propagation} and \emph{decisions} are used by CDCL to extend $D$.
Boolean constraint propagation adds implied literals to $D$, where a literal
is implied if there exists a clause where this literal is the only one
remaining that can be satisfied by an extension of $D$, as is checked in Line~\ref{alg:cdcl:bcp}.
Decisions add a selected literal to $D$ after exhaustively applying Boolean constraint propagation
without a conflict as denoted in Line~\ref{alg:cdcl:decide}.

This description of CDCL is deliberately conceptual. 
Modern SAT and MaxSAT solvers are well engineered pieces of software
that use sophisticated data structures and algorithms which are integral
to their performance.
Particularly conflict analysis, backtracking, and decisions have not been covered here.
We refer the interested reader to~\cite{sathandbook2021} for more details.

With the knowledge of how CDCL works, we discuss the integration of cycle propagation 
next.

\subsection{Cycle Propagation}
Conflicts are a central concept in CDCL, as they signal the solver that
a partial assignment is infeasible.
Cycle propagation uses this mechanism to ensure that the solver stops 
as soon as $D$ implies a cycle.
We perform this check after Boolean constraint propagation in Line~\ref{alg:cdcl:dfvs}.

Cycles are only implied by negative literals, since negative literals indicate that a vertex
remains in the graph.
Hence, it is sufficient to check if the negative literals
$V' = \SB v \in V \SM \lnot v \in D \SE$ induce an acyclic graph, i.e., if $G[V']$ is acyclic.
In case a cycle $C$ is found, it is added as a clause to $\mathcal{C}$.

Adding $C$ immediately causes a conflict, since by definition~$D$ contains the negation of~$C$.
Hence, we achieve our goal of immediately stopping the solver.
Further, we add a single cycle and corresponding conflict clause, thereby minimizing
the number of extra constraints.
This usually guides the solver quicker to a feasible solution 
than adding several cycles after the solver returns an infeasible solution.
Note that the solver with cycle propagation
never returns an infeasible solution.

The acyclicity check is performed using a DAG implemented as a simple doubly linked data structure.
Here, each vertex knows its predecessors, successors and has an order.
The structure preserves two invariants: it is a DAG, and the order of a vertex is 
the maximum order over its predecessors plus one, or
$0$ if the vertex has no predecessors.
Whenever a new vertex is inserted, its order is recursively propagated to the successors.
Recursive calls are only necessary, if the propagated order plus one is larger than the successor's order.
Should the propagation reach the inserted vertex, we have found a cycle
and we remove the vertex, preserving the invariant that the structure is a DAG.
Removal of a vertex requires recursively propagating the change to all successors
whose ordering depends on the target vertex.

Cycle propagation is performed after Boolean constraint propagation for practical reasons.
First, modern SAT solvers spend most of their time performing Boolean constraint propagation
and can perform this task very fast. 
Checking for cycles after each change to $D$ would, therefore, cause a considerable slowdown
of the solver.
Second, we keep track of the changes to the partial assignment in between cycle propagation runs.
This allows us to perform the aforementioned modifications to our data structure in bulk, further speeding up the acyclicity check.
With these considerations, the runtime percentage dedicated to cycle propagation shown in the profiler
is in the low single digits
as Boolean constraint propagation still takes up almost all of the
runtime.

This concludes the conceptual description of our approach.
Next, we discuss our empirical evaluation.

    \section[Experiments]{Experiments\footnote{Results and source code are available at \url{https://doi.org/10.5281/zenodo.7307445}.}}\label{sec:experiments}
\paragraph{Instances}
We use instances from the recent \emph{PACE}, the argumentation framework competition \emph{ICCMA}, and \emph{random} graphs.
The recent PACE provides 200 dedicated DFVSP instances.\footnote{\url{https://pacechallenge.org/2022/tracks/}}\footnote{At the time of writing, details on the origin of the instances have not been released.}
The 137 ICCMA instances come from a recent argumentation framework competition\footnote{\url{https://argumentationcompetition.org/2021/}},
where we selected those instances with 50 to 1000 vertices.
We also generated 1140 random instances using different parameterizations for 
the number of vertices and the probability an edge exists using the methodology of~\cite{zhou2016spin}.
We generated 10 instances for each parameterization, which
are reported as 114 instances, averaging the results over the respective 10 instances.
The number of vertices ranges between 100 and 10000 and the average degree of a vertex
varies from 2 to 50.

We preprocessed the instances by removing all self-loops, as the PACE instances met this requirement and the competition solvers were not able to deal with instances containing self-loops.

\paragraph{Implementation}
We implemented the proposed algorithm in our solver \DAGger\footnote{\url{https://github.com/ASchidler/dfvs}}.
Our implementation is based on the MaxSAT solver 
EvalMaxSAT~\cite{avellaneda2020}, which uses Glucose 3 in the backend~\cite{audemard2009}.
We chose EvalMaxSAT because it placed well in the 2021 MaxSAT evaluation\footnote{\url{https://maxsat-evaluations.github.io/2021/}} and the code base
has no dependencies and can easily be modified and integrated.

We initially give up to $25000$ short cycles with a maximum length of $4$ 
to the MaxSAT solver.
These limits have performed best overall.
A lower maximum length does not find any cycles for some instances, while
a higher maximum length seems to slow down the solver, as does a larger number of cycles.

\paragraph{Setup}
Our implementation uses C++ and was compiled using gcc~7.5.0.
We compared our solver to the second place PACE solver \emph{grapa-java}\footnote{\url{https://gitlab.informatik.uni-bremen.de/grapa/java}}, which uses a CEGAR-like approach together with an integer linear programming solver and new data reductions.%
\footnote{We tried to obtain further solvers for comparison.
Unfortunately, for~\cite{bao2018analysis}, we did not manage to get in contact with the authors, for~\cite{fages2006constraint} the source code is lost and the implementation of~\cite{koehler2005contraction}
did not contain an exact solver}
As an additional baseline we also used a direct DFVSP encoding into \emph{SAT},
based on the transitive closure encoding for acyclicity~\cite{janota2017}, 
using our data reductions for preprocessing.

We used a time limit of 30 minutes and a memory limit of 8~GB.
The experiments were run on servers with two AMD EPYC 7402 CPUs, each with 24 cores running at 2.8 GHz, and using Ubuntu~18.04.

An instance counts as \emph{solved}, if it was solved in all five runs,
otherwise if it was solved in at least one run, it is counted as \emph{partially solved}.
The given values are averaged over all runs.

\subsection{Solvers}
Comparing \DAGger's performance to that of other solvers, was the goal of
our first experiment.
The results, together with different configurations from the next experiment,
are shown in Table~\ref{tab:results-overall} and 
as a cactus plot in Figure~\ref{fig:results-cactus}.

\begin{table}[htb!]
\centering
	\begin{tabular}{l rr rr rr } 
	\toprule
& \multicolumn{2}{c}{PACE} & \multicolumn{2}{c}{ICCMA} & \multicolumn{2}{c}{Random} \\
	Solver & S & P & S & P & S & P \\ 
	\cmidrule(r){1-1}
\cmidrule(lr){2-3}
\cmidrule(lr){4-5}
\cmidrule(lr){6-7}
	\DAGger &186&2	&64&0	&13&5 \\
	grapa-java&165&0	&52&0	&10&5 \\
	SAT&134&3	&59&0	&12&5	 \\
\cmidrule(r){1-1}
\cmidrule(lr){2-3}
\cmidrule(lr){4-5}
\cmidrule(lr){6-7}
Configuration& S & P & S & P & S & P  \\
\cmidrule(r){1-1}
\cmidrule(lr){2-3}
\cmidrule(lr){4-5}
\cmidrule(lr){6-7}
No CP &180&2	&62&0	&11&6	 \\
No DR&151&6	&63&1	&13&5	 \\
No CP \& DR &146&3	&62&0	&10&7	 \\
\bottomrule
	\end{tabular}
	\caption{Number of solved instances for different solvers and \DAGger configurations.
	\emph{S} and \emph{P} show the number of solved and partially solved instances respectively. 
	}
	\label{tab:results-overall}
\end{table}

\begin{figure}[htb!]
\centering
	\includegraphics[trim={4mm 2mm 0 0},clip, scale=0.59]{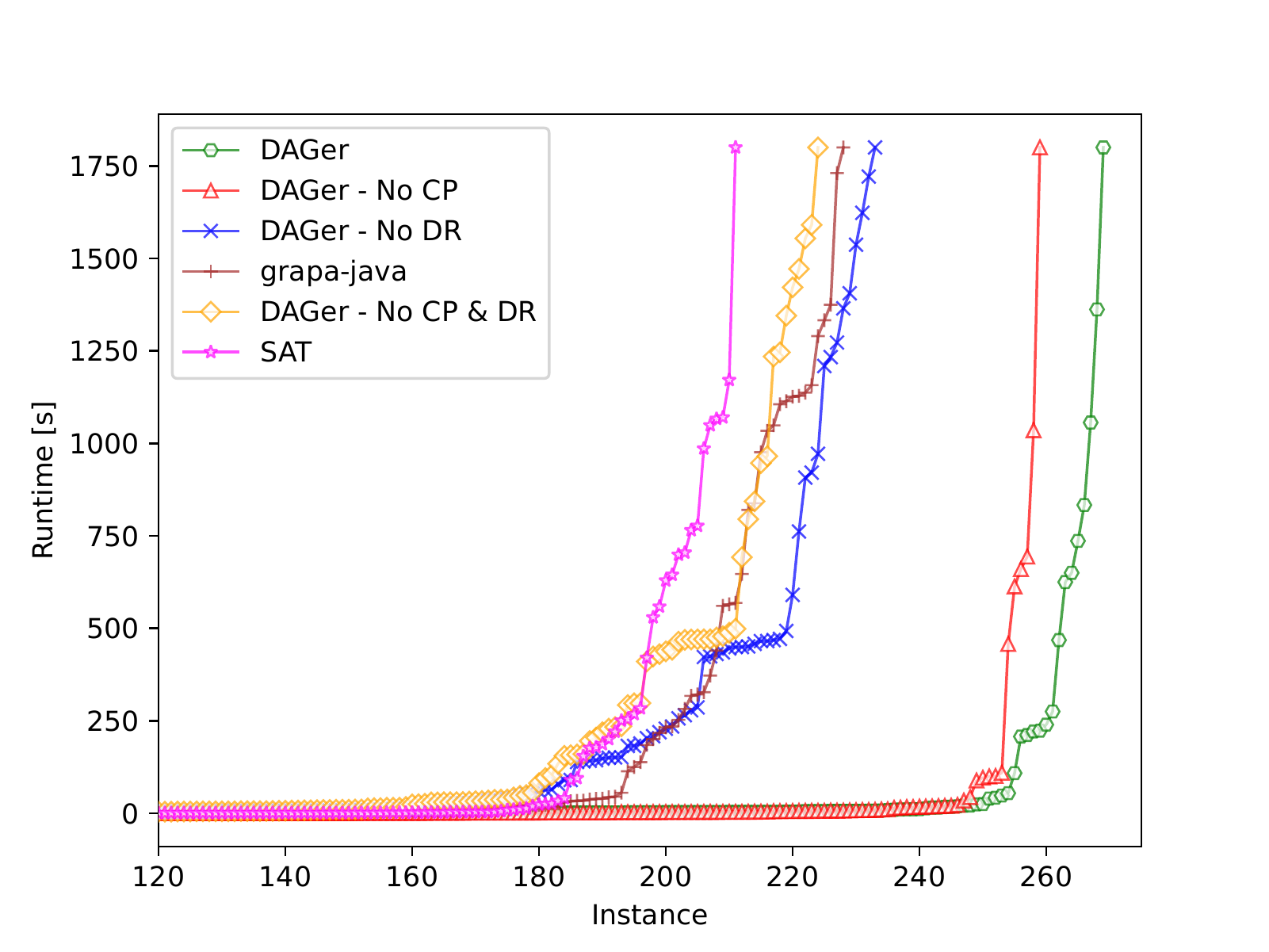}
	\caption{Cactus plot for different solvers and \DAGger configurations.}
	\label{fig:results-cactus}
\end{figure}

\DAGger performs better than both grapa-java and the SAT encoding for every
instance group.
Interestingly, the SAT encoding performs better than grapa-java on non-PACE
instances. 
The cactus plot shows that instances are either very hard, or very easy,
with very few solved instances having a high runtime.

\DAGger excels on the PACE instances and solved almost half the ICCMA instances,
but random instances seem to be hard for all solvers.
We will examine this behavior for \DAGger further in subsequent experiments.

\subsection{Features}
We measured the impact of our contributions by
disabling cycle propagation (CP) or our new data reductions (DR).
\Cref{tab:results-overall} shows the performance of these three additional configurations
and \Cref{fig:results-cactus} shows the runtime behavior.
Without cycle propagation, \DAGger calls
the MaxSAT solver incrementally and whenever the solution $D$ is infeasible,
we add disjoint cycles from $G - D$.
Hence, our experiment tests precisely the benefit 
the integration into the solver.

Cycle propagation has a small impact in terms of the number of instances.
While the number is small, the respective instances are hard and contain a very large 
number of uncovered cycles that we were unable to enumerate within the runtime.
Without cycle propagation, \DAGger solved 253 instances.
The initial solution was infeasible for 132 of those instances, lazily generated
clauses were not necessary for the remaining instances.

The impact of the data reduction depends strongly on the instance set.
PACE instances are heavily reduced, but the impact on ICCMA and random instances
is almost none.
The reason for this is that our new reductions rely heavily on structural properties
that are unlikely in randomized graphs.
For ICCMA instances, the reason is different, which we will explore next.

Overall, without our contributions, \DAGger would perform worse than grapa-java, but
better than the SAT baseline, showing that the incremental approach is the more promising
SAT approach for DFVSP.

\subsection{Data Reductions}
The effectiveness of the data reductions is not well represented by
the number of instances the solving algorithm can solve,
as in the future they might be beneficial for instances that are too hard for current solvers.
Figure~\ref{fig:dr-absolute} shows how much all instances were reduced in size and
offers some interesting insights.
First, many instances, particularly ICCMA instances, are directly solved by
our data reductions. 
The ICCMA instances seem to be either easily reducible, in which case they are also easy
to solve, or they are hard to reduce and solve.
Second, the result on random instances shows that the reductions become less effective with 
increasing density.
Lastly, on most instances, the data reductions can significantly decrease the instance's size.
\begin{figure}[htb!]
    \centering
    \includegraphics[trim={5mm 0 0 0},clip, scale=0.59]{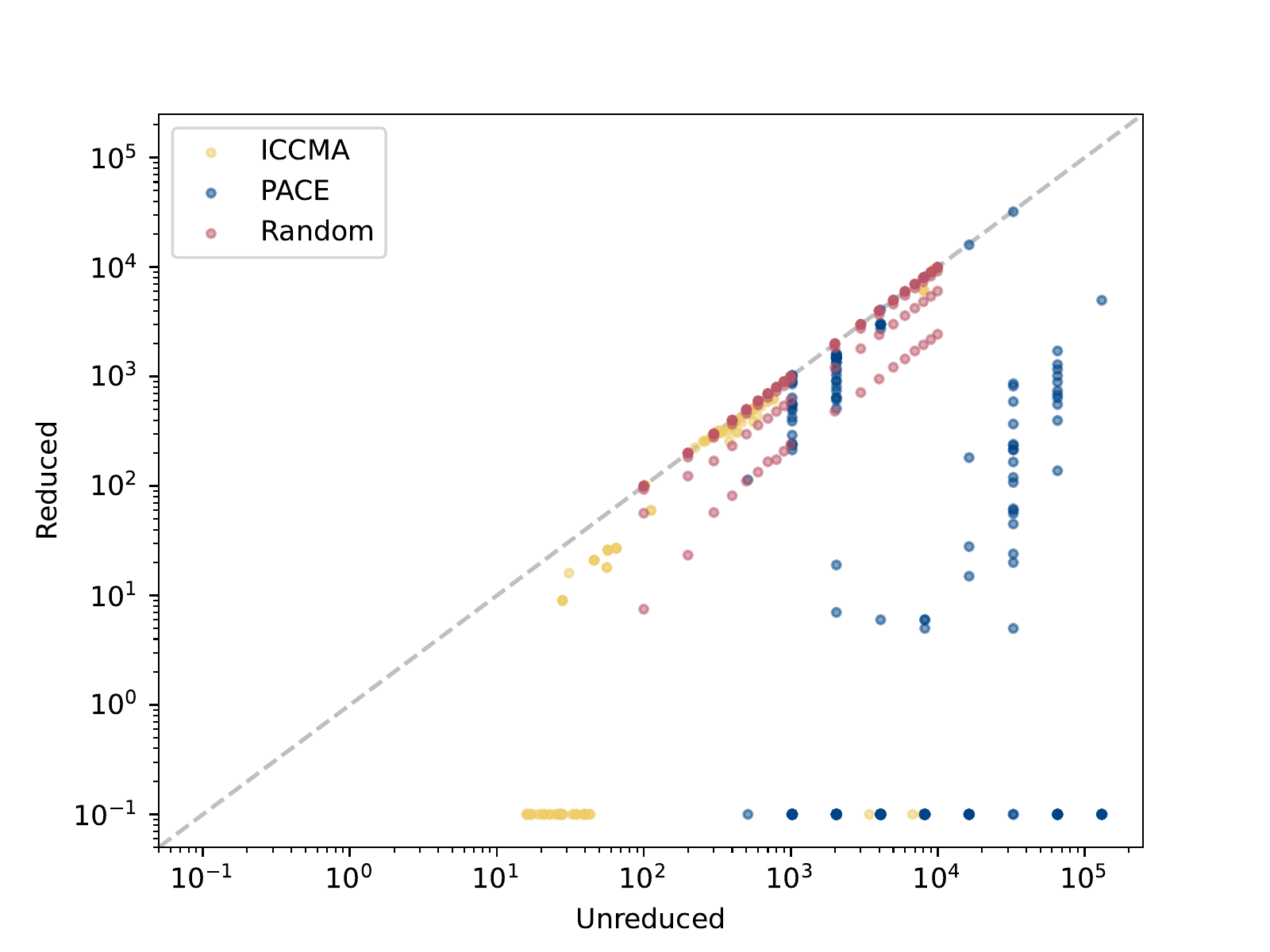}
    \caption{Comparison of graph sizes before and after applying data reductions. Due to the use of the logarithmic scale, we treat $0$ as $0.1$.}
    \label{fig:dr-absolute}
\end{figure}

We also wanted to see how much benefit our computationally more expensive reductions
have over the simple reductions proposed by~\citet{levy1988contraction}.
Figure~\ref{fig:dr-cmp} shows that our data reductions do not have much benefit over
the simple reductions on random instances.
For PACE, instances the reductions are very useful, reducing the size of almost all instances,
directly solving many of them.
For the ICCMA instances, the reductions either solve the instance or are ineffective.

\begin{figure}[htb!]
    \centering
    \includegraphics[trim={5mm 0 0 0},clip, scale=0.59]{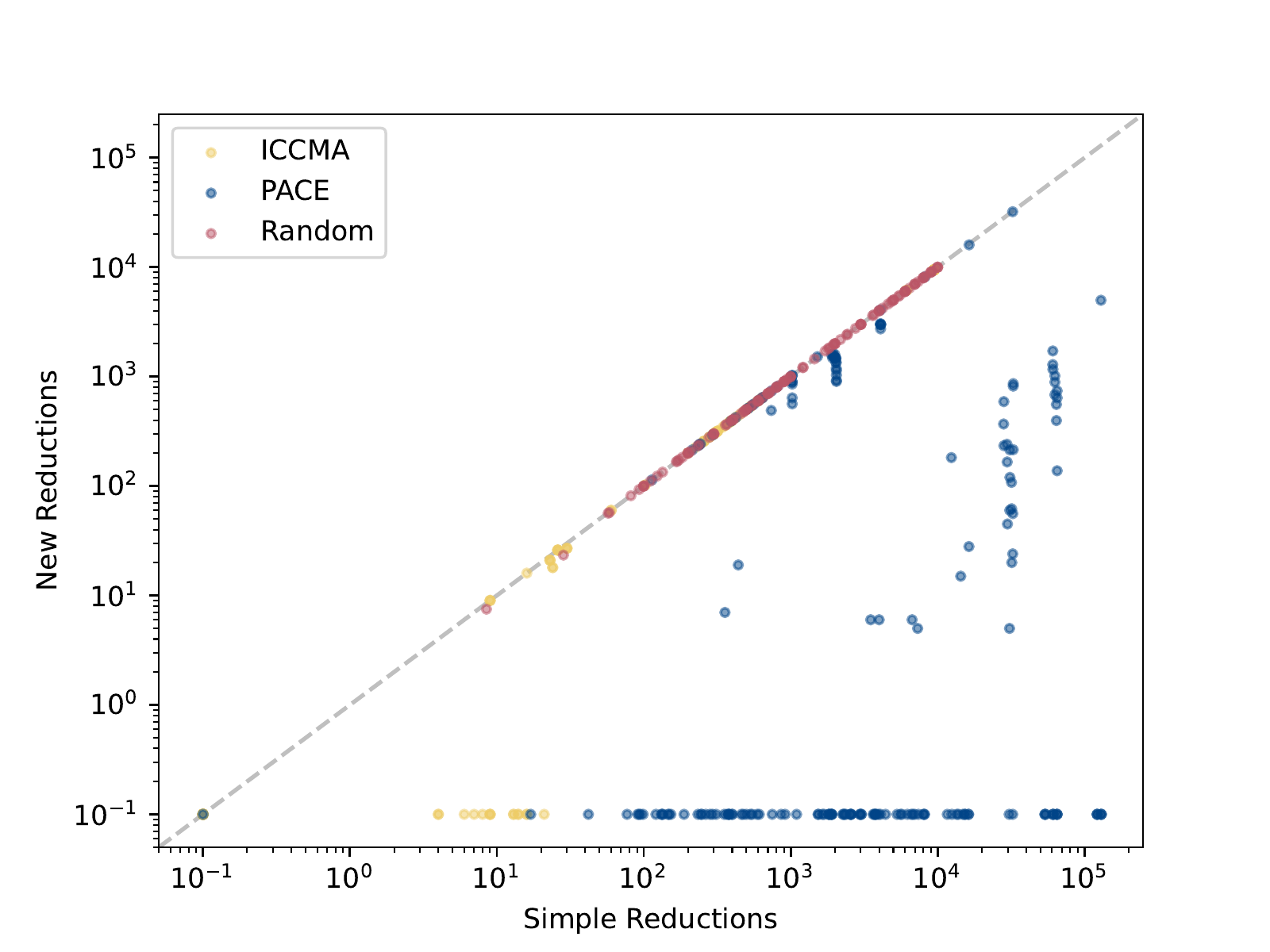}
    \caption{Comparison between using only simple reductions and using
    the data reductions we propose.
    Due to the use of the logarithmic scale, we treat $0$ as $0.1$.}
    \label{fig:dr-cmp}
\end{figure}

\subsection{Instance Size}
The potential correlation between the graph's size and the MaxSAT solver's
ability to find a minimum DFVS, was the focus of our last experiment.
Figure~\ref{fig:instances} shows which instances have been solved, partially solved, or
remained unsolved, in relation to the number of vertices and density.
While the number of uncovered cycles would also have been of interest,
we could not enumerate them in a reasonable amount of time for hard instances.
Since we were interested in the MaxSAT solver's performance,
the figure uses the data from the reduced instances.

\begin{figure}[htb!]
    \centering
    \includegraphics[trim={5mm 0 0 0},clip, scale=0.59]{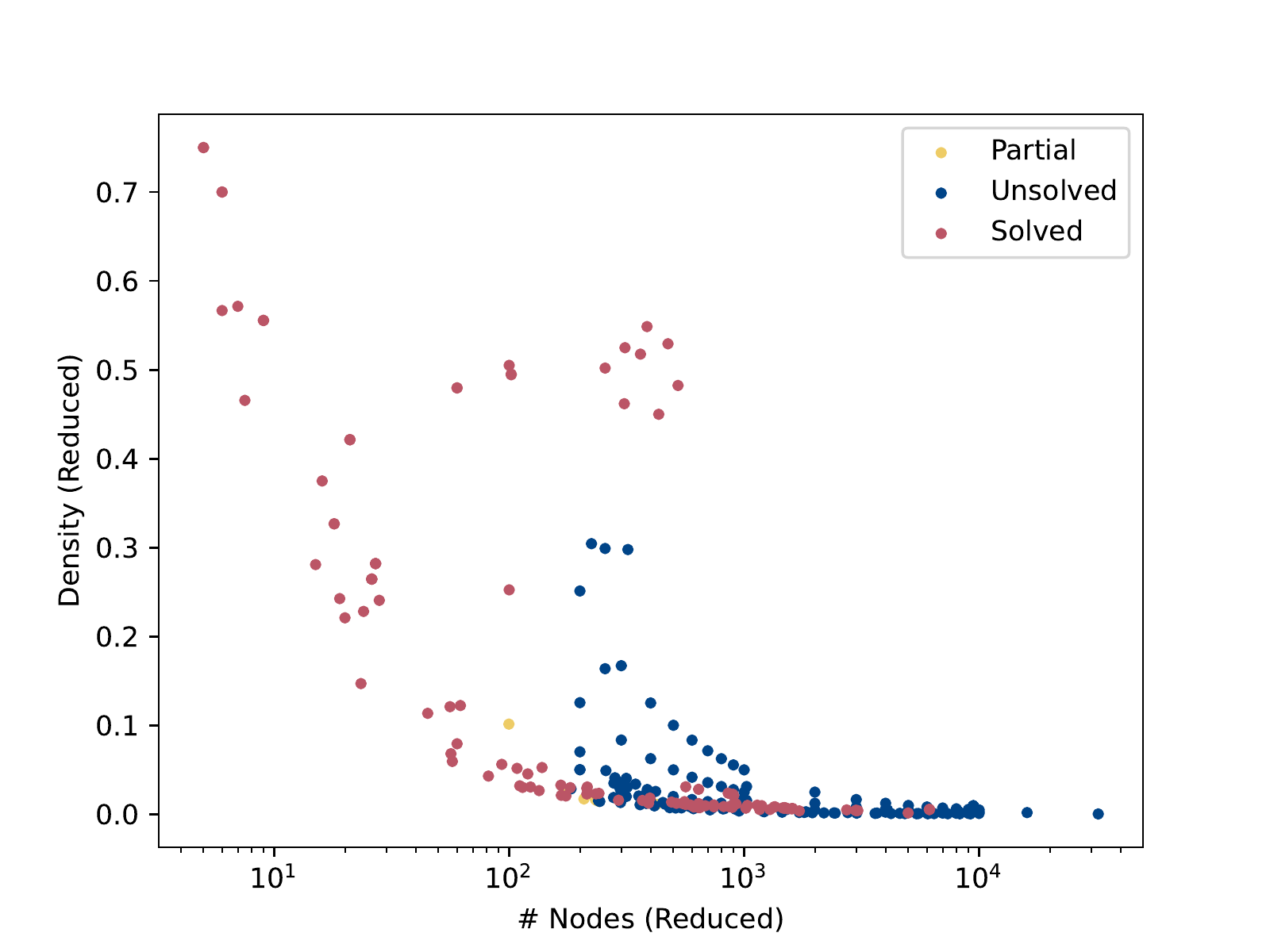}
    \caption{Solved instances in relation to the graph size and the density.}
    \label{fig:instances}
\end{figure}

The figure shows that increased size and density indeed make the instance harder.
Whereas for small graphs up to around 100 vertices the density does not matter much,
this changes for larger graphs.
The size limit for our approach seems to be around 10000 vertices, where the
solver fails even for very sparse graphs.

Interestingly, at around 1000 vertices there is a cluster of instances with high density that the
solver solved successfully. 
It seems that instances where almost the whole graph is part of a minimum
DFVS are again easier to solve than mid-density instances.

Random instances provide some more insight.
For 100 vertices, \DAGger solved almost all instances.
The only exception is when the minimum DFVS size is around 50, here \DAGger only managed to solve half the instances, hence, they seem to be harder.
For the remaining instances, \DAGger was not able to solve instances with average degree higher
than 3, but managed to solve random instances with up to 1000 vertices.


\subsection{Discussion}
The results show that the data reductions perform particularly well on the PACE instances,
where many instances have a high enough number of bi-edges.
While still useful on the ICCMA instances, as there they solve many instances directly, they are
not necessary, as the solver would have solved almost all of them.

Cycle propagation works in a complementary fashion to our data reductions.
It works particularly well on instances that have few uncovered short cycles, 
but a large number of uncovered cycles.
While not many of the instances in our instance sets fell into this category, cycle propagation
helped to solve several hard instances. 

In general, a CEGAR approach works well for DFVSP as
even without
cycle propagation and our data reductions,
the solver outperformed the second best PACE solver on non-PACE instances.

Particularly challenging for all tested solvers are instances of high edge density,
although there is a visible trend that indicates that instances with very high density
could in turn become easier.

\section{Conclusion}
In this paper, we discussed our novel approach to DFVSP.
Key features are new data reductions lifted from related problems. 
Apart from the reductions themselves, we also provided a theoretical basis
that can be used to lift further reductions in the future.
The other key feature is cycle propagation. 
While lazily extending the set of constraints to obtain a feasible solution with a limited set of constraints works well,
we managed to solve several hard instances by integrating this extension directly into the MaxSAT solver.



There are more data reductions from VCP that we did not consider for DFVSP, since they are based on very non-local conditions or seem highly difficult to adapt to DFVSP such as the CROWN~\cite{abu2004direct} or LP-based reduction~\cite{nemhauser1975vertex}. We nevertheless hope to incorporate more reduction techniques. For local VCP reductions Theorem \ref{thm:dfvs_boundary} is a strong methodological foundation for such a transfer, however, we hope to expand this and possible incorporate ideas from other related problems.

We have shown that cycle propagation works well in practice.
We think that there are two avenues where we might further improve its performance:
(i) there are several SAT solver details that might be used to further improve performance,
particularly adapting inprocessing and decision heuristics to incorporate domain specific
knowledge about DFVSP, and
(ii) we used a core-guided MaxSAT solver for our implementation;
it would be interested to see how well cycle propagation performs integrated into an
implicit hitting set based MaxSAT solver.

\av{
\appendix
\section{Standard DFVSP Reductions}
\label{app:reductions}

Here, we use $G\circ v$ as the digraph, called the exclusion of $v$ from a digraph $G$ by letting $G \circ v := G - v + N_{succ}(v)\times N_{pre}(v)$.
The most well-known reduction rules for DFVSP preprocessing are those of~\citet{levy1988contraction}:
\begin{reduction}[LOOP]
If there exists $v \in V(G)$ such that $(v,v) \in E(G)$ replace $G$ by $G - v$.
\end{reduction}
\begin{reduction}[IN0/1]
If there exists $v \in V(G)$ such that $v$ has at most one incoming edge, replace $G$ by $G \circ v$.
\end{reduction}
\begin{reduction}[OUT0/1]
If there exists $v \in V(G)$ such that $v$ has at most one outgoing edge, replace $G$ by $G \circ v$.
\end{reduction}

The latter two rules were later subsumed by~\citet{lemaic2008markov}, using the reductions.
\begin{reduction}[INDICLIQUE] 
If there exists $v \in V(G)$ such that the incoming edges of $v$ form a diclique, replace $G$ by $G \circ v$.
\end{reduction}
\begin{reduction}[OUTDICLIQUE] 
If there exists $v \in V(G)$ such that the outgoing edges of $v$ form a diclique, replace $G$ by $G \circ v$.
\end{reduction}
Apart from this, \citeauthor{lemaic2008markov} introduced two new reductions
\begin{reduction}[DICLIQUE-2] 
If there exists $v \in V(G)$ whose neighbors can be partitioned into two disjoint cliques $N_1, N_2$ such that the bi-edges of $v$ are a strict subset of $N_1$, replace $G$ by $G \circ v$.
\end{reduction}
\begin{reduction}
[DICLIQUE-3] If there exists $v \in V(G)$ without bi-edges whose neighbors can be partitioned into three disjoint cliques $N_1, N_2, N3$, replace $G$ by $G \circ v$.
\end{reduction}
Note, that for soundness of all the above reductions it is necessary that LOOP is not applicable to $v$.

Furthermore, \citeauthor{lin2000computing} introduced three further reductions that make use of bi-edges in the digraph. 
\begin{reduction}[PIE] 
If there is an arc $(u,v) \in E(G)$ such that $(v,u) \not\in E(G)$ and every path from $v$ to $u$ in $G$ uses a bi-edge, replace $G$ by $G - (u,v)$.
\end{reduction}
\begin{reduction}[DOME] 
If there is an arc $(v,u) \in E(G)$ such that $(u,v) \not\in E(G)$ and one of the following holds
 \begin{itemize}
    \item $\SB p \SM (p, v) \in E(G), (v,p) \not \in E(G)\SE \subseteq \SB p \SM (p,u) \in E(G) \SE$, i.e., for every $(p,v) \in E(G)$ that is not a bi-edge there is an arc $(p,u) \in E(G)$. 
    \item $\SB p \SM (u, p) \in E(G), (p,u) \not \in E(G)\SE \subseteq \SB p \SM (v,p) \in E(G)\SE$, i.e., for every $(u,p) \in E(G)$ that is not a bi-edge there is an arc $(v,p) \in E(G)$. 
\end{itemize}
then replace $G$ by $G - (v,u)$.
\end{reduction}
\begin{reduction}[CORE] 
If there exists $v \in V(G)$ such that all arcs of $v$ are bi-edges and the neighbor of $v$ form a diclique, replace $G$ by $G \circ v$.
\end{reduction}
Note that also the CORE reduction is a special case of the INDICLIQUE and OUTDICLIQUE reductions.
\section{Standard VCP Reductions}
\label{app:vc-reductions}
\begin{reduction}[SUBSET~\cite{stege1999improved}]
If there exists $v,u \in V(G)$ such that $\{v,u\} \in E(G)$ and $N(v) \subseteq N(u) \cup \{u\}$, then replace $G$ by $G-u$.
\end{reduction}


Another reduction by \citeauthor{fellows2018what} is more complicated but generalizes many others like the 2FOLD reduction~\cite{xiao2013confining}.
\begin{reduction}[MANYFOLD {{\cite{fellows2018what}}}]
If there exists a vertex $v \in V(G)$ such that there is a partition
$(C_1,C_2)$ of $N(v)$, where
 \begin{itemize}
    \item $|C_1| \geq |C_2|$,
    \item $C_i$ is a clique for $i = 1,2$, and
    \item for each $c_1 \in C_1$, there is precisely one $c_2 \in C_2$ such that $\{c_1,c_2\} \not\in E(G)$.
\end{itemize}
Then, replace $G$ by 
\[
G - v - C_2 + \bigcup_{\{c_1, c_2\} \in M, c_1 \in C_1} \{c_1\} \times N(c_2),
\]
where $M$ denotes the set of missing edges from $G[N(v)]$.
\end{reduction}

Whereas MANYFOLD works well on dense graphs, the following reduction works on sparse graphs.
\begin{reduction}[4PATH~\cite{fellows2018what}]
If there exists a vertex $v \in V(G)$ such that 
 \begin{itemize}
    \item $N(v) = \{a,b,c,d\}$, and
    \item $E(G[N(v)]) = \{\{a,b\},\{b,c\},\{c,d\}\}$,
\end{itemize}
then replace $G$ by 
\begin{align*}
G   &- v + \{\{a,c\}, \{a,d\}, \{b,d\}\} \\
    &+ \{a,b\}\times N(d) + \{c,d\}\times N(a). 
\end{align*}
\end{reduction}
\begin{algorithm}[t]
\begin{algorithmic}[1]
\algnewcommand{\algorithmicgoto}{\textbf{go to}}%
\algnewcommand{\Goto}[1]{\State \algorithmicgoto~#1}%

\Function{CheckUnconfined}{$v, G$}
\State $S \leftarrow \{v\}$
\State $P \leftarrow  \SB u \in N(S) \SM |N(u) \cap S| = 1\SE$
\If{$P = \emptyset$}
    \State $u \leftarrow  \text{argmin}_{u' \in P} |N(u') \setminus (N(S) \cup S)|$
    \If{$|N(u') \setminus (N(S) \cup S)| = 0$}
        \Return True
        \ElsIf{$|N(u) \setminus (N(S) \cup S)| = 1$}
        \State $S \leftarrow S \cup \{u\}$
        \Goto 3
    \EndIf
\EndIf
\Return False
\EndFunction
\end{algorithmic}
 \caption{An algorithm that checks whether a vertex $v$ is unconfined in a graph $G$.}
\label{alg:unconfined}
\end{algorithm}

While the above reductions all capture a fixed graph pattern, the following applicability of the following one is determined by an iterative procedure in~\Cref{alg:unconfined}.
\begin{reduction}[UNCONFINED~\cite{xiao2013confining,akiba2016branch}]
If there is a vertex $v \in V(G)$ such that \textsc{CheckUnconfined}($v,G$), replace $G$ by $G - v$.
\end{reduction}
All of these reductions induce (many) boundary reductions.
\section{Proofs}\label{sec:proofs}

\begin{theorem}
For every graph $G$ such that $r = \langle H, H', B, c\rangle$ is applicable
it holds that for every minimum vertex cover $S$ of $G + H$ there is a minimum vertex cover $S'$ of $G + H'$ such that
 \begin{itemize}
    \item $|S| = |S'| + c$, and
    \item $S \cap V(G) = S' \cap V(G)$.
\end{itemize}
\end{theorem}
\begin{proof}
So let $S$ be a minimum vertex cover of $G + H$ and $X = B \cap S$. We know that $VC(H - X) = VC(H' - X) + c$ and that $S_l = S \cap (V(H) \setminus B)$ is a minimal vertex cover of $H - X$. Therefore, $|S_l| = VC(H-X)$, which implies that there exists a vertex cover $S_l'$ of $H' - X$ such that $|S_l| = |S_l'| + c$. Then, $S' = S \cap V(G) \cup S_l'$ is a vertex cover of $G + H'$ and it holds that $|S| = |S'| + c$ and $S \cap V(G) = S' \cap V(G)$. 

It remains to show that $S'$ is also minimum. Assume that there was another vertex cover $C$ of strictly smaller cardinality. Then we can use the same steps as above to obtain a vertex cover $C'$ of $G + H$ of strictly smaller cardinality than $S$. This is a contradiction, which implies that $S'$ is a minimum vertex cover.
\end{proof}
\begin{theorem}
Let $G$ be a digraph and $r = \langle H, H', B, c\rangle$ be a boundary VCP reduction and $\Pi(G) = G' + V(H) + E(H)$ such that $r$ is applicable to $G'$. If every vertex $v \in V(H) \setminus B$ only has bi-edges in $G$ and every arc $(u,w) \in E(G)$ is a bi-edge or $|\{u,w\} \cap B| \leq 1$,
then 
\[
DFVS(G) = DFVS(G^*) + c,
\]
where $G^*$ is given by
\begin{align*}
G & - (V(H) \setminus B) - B \times B\\
  & + V(H') + \SB (u,v) \SM \{u,v\} \in E(H')\SE.
\end{align*}
\end{theorem}
\begin{proof}
The proof is analogous to that of Theorem~\ref{thm:vc_boundary}.
\end{proof}
\begin{theorem}
Let $G$ be a loop-free digraph, such that MANYFOLD is applicable to $v^{*} \in V(G)$ and $G'$ be the graph obtained from $G$ after MANYFOLD was applied on vertex $v^*$, then 
 \begin{itemize}
    \item $DFVS(G) = DFVS(G') + |C_2|$,
    \item and given a minimum DFVS of $G'$, we can in polynomial time compute a minimum DFVS of $G$.
\end{itemize} 
\end{theorem}
\begin{proof}
The main observation that is used to proof soundness for the MANYFOLD reduction in the VCP case,
is that without loss of generality there are only two possible cases we need to consider. For this, first note that if $v^{*}$ is not contained in a minimum DFVS of $G$, then the vertices in $C_1 \cup C_2 = N^{G}_{bi}(v^{*})$ are. On the other hand, if $v^{*}$ is in a minimum DFVS $S$ of $G$, then still, since $C_1$ and $C_2$ are dicliques, it holds that $|S \cap C_i| \geq |C_i| - 1$ for $i = 1,2$. In fact, if $|S\cap C_i| = |C_i|$ for one of $i = 1,2$ then we can assume that it holds for both $i = 1$ and $i = 2$, since in this case $v^*$ must be in $S$ and we can replace it by the missing vertex without changing the size. Therefore, w.l.o.g. for a minimum DFVS $S$ of $G$ it holds that either
\begin{enumerate}
    \item $S \cap C_i = C_i$ for $i = 1,2$ and $v^{*} \not\in S$, or
    \item $|S \cap C_i| = |C_i| - 1$ for $i = 1,2$.
\end{enumerate}

Assume now, that we are given a minimum DFVS $S$ of $G$ such that 1.\ holds. In this case, we obtain a DFVS of $G'$ as $S' = S \setminus C_2$ and it follows that $DFVS(G) - |C_2| \geq DFVS(G')$. To see that $S'$ is a DFVS of $G'$ observe that every edge that $G'$ has but not $G$ contains a vertex from $C_1 \subseteq S \setminus C_2 = S'$.

If instead we are given a minimum DFVS $S$ of $G$ such that 2.\ holds, then let $c_i \in C_i \setminus S$. In this case, $S' = S \setminus (C_2 \cup \{v^*\})$ is a DFVS of $G'$ and thus $DFVS(G) - |C_2 \setminus \{c_2\}| + |\{v^*\}| = DFVS(G) - |C_2| \geq DFVS(G')$. To see that $S'$ is a DFVS of $G'$, assume the contrary. Then there must be cycle that is not covered by $S'$. 


We know that at at least one of $(c_1,c_2)$ and $(c_2,c_1)$ is in $M$ and proceed by a case distinction on whether both are in $M$ or not.

Case both are in $M$: Then there is no uncovered path between $c_1$ and $c_2$ since the reduction is applicable. Assume there is a (w.l.o.g.) uncovered cycle. Then this cycle must use the vertex $c_1$, since the only added arcs, which do not have a vertex in $S'$ use $c_1$. Furthermore, the cycle must contain an arc between $N^G(c_2)$ and $c_1$ and between $N^G(c_1)$ and $c_1$. If only the latter holds true, then the same cycle is also present in $G-S$. If only the former holds true, then there is an equivalent cycle in $G-S$, where $c_1$ is replaced by $c_2$. Since the cycle is uncovered, the arcs must go into opposite direction, i.e., be of the form $(v_2,c_1)$ and $(c_1, v_1)$ or $(v_1,c_1)$ and $(c_1, v_2)$, where $v_i \in N^G(c_i), i = 1,2$. Assume the first form, then we can transform the uncovered cycle into an uncovered path from $c_1$ to $c_2$ in $G$ by going from $v_2$ to $c_2$ instead of $c_1$. This is a contradiction to the assumption that there are no uncovered paths between $c_1$ and $c_2$. The argument for the latter form is analogous.

Case only $(c_1,c_2)$ is in $M$.
Thus, every uncovered path in $G$ from $c_2$ to $c_1$ uses the arc $(c_2,c_1)$. This implies that a cycle $c_1\dots c_1$ in $G' - S'$ is also a cycle in $G - S$, there is a corresponding cycle with $c_1$ replaced by $c_2$ in $G$, or there is a corresponding cycle $c_1\dots c_2c_1$ in $G-S$. This a contradiction to the assumption that $S$ is DFVS of $G$.

Case only $(c_2,c_1)$ is in $M$.
Thus, every uncovered path in $G$ from $c_1$ to $c_2$ uses the arc $(c_1,c_2)$. This implies that a cycle $c_1\dots c_1$ in $G' - S'$ is also a cycle in $G - S$, there is a corresponding cycle with $c_1$ replaced by $c_2$ in $G$, or there is a corresponding cycle $c_1c_2\dots c_1$ in $G-S$. This a contradiction to the assumption that $S$ is DFVS of $G$.

Thus, it follows that $DFVS(G) \geq DFVS(G') + |C_2|$. 

As for the other direction, let $S'$ be a minimum DFVS of $G'$. Again, we consider two cases, namely $|S' \cap C_1| = |C_1|$ and $|S' \cap C_1| = |C_1|-1$, which are the only possible cases for the cardinality of the intersection.

Case $|S' \cap C_1| = |C_1|$: Then $S = S' \cup C_2$ is a DFVS of $G$, which implies $DFVS(G) \leq DFVS(G') + |C_2|$. To see that $S$ is a DFVS of $G$, observe that $G - S - v^{*}$ and $G' - S'$ are equal and $N(v^{*}) \subseteq C_1 \cup C_2 \subseteq S$.

Case $|S' \cap C_1| = |C_1|-1$: Let $c_1 \in C_1 \setminus S'$ and $c_2 \in C_2$ the unique vertex such that there is no bi-edge between $c_1$ and $c_2$ in $G$. Then $S = S' \cup (C_2 \setminus \{c_2\}) \cup \{v^{*}\}$ is a DFVS of $G$, which implies $DFVS(G) \leq DFVS(G') + |C_2|$. It remains to show that $S$ is a DFVS of $G$. 
We consider two subcases for the number of arcs that use $c_1$ and $c_2$ in $M$.


Case only one of $(c_1,c_2)$ is in $M$ or $(c_2,c_1)$ is in $M$. Assume $(c_1,c_2)$ is in $M$, the other case works analogously.
Thus, every uncovered path in $G$ from $c_2$ to $c_1$ uses the arc $(c_1,c_2)$. Assume that there is a (w.l.o.g.) uncovered cycle in $G-S$. We know that this cycle must use $c_1$ and is thus of the form $c_1\dots c_1$. Furthermore it must use $c2$, which implies that there is an uncovered path from $c_2$ to $c_1$ as a part of the cycle. This the cycle is actually of the form $c_1\dots c_2c_1$.
Then however, there is a corresponding cycle $c_1\dots c_1$ in $G'-S'$, since all the arcs of $c_2$ were added to $c_1$, which is a contradiction.

Case both $(c_1,c_2)$ and $(c_2, c_1)$ are in $M$. Assume that there is a (w.l.o.g.) uncovered cycle in $G-S$. As in the previous case, we know that the cycle must use both $c_1$ and $c_2$, since it is otherwise also a cycle in $G' - S'$. Thus, this cycle gives us an uncovered path from $c_1$ to $c_2$, which is a contradiction to the assumption on $M$.

Since these are the only cases, we are done and $DFVS(G) \leq DFVS(G') + |C_2|$, meaning that overall $DFVS(G) = DFVS(G') + |C_2|$. Furthermore, the constructions used in the proof are possible in polynomial time, which was the second claim of the theorem.
\end{proof}
\begin{theorem}
Let $G$ be a digraph such that 4PATH is applicable to $v^* \in V(G)$ and $G'$ be the graph obtained from $G$ after 4PATH was applied to the vertex $v^*$, then
 \begin{itemize}
    \item $DFVS(G) = DFVS(G')$, and
    \item given a minimum DFVS of $G'$, we can in polynomial time compute a minimum DFVS of $G$.
\end{itemize}
\end{theorem}
\begin{proof}
Let $S$ be a minimum DFVS of $G$. If $N^G(v^*) \subseteq S$, then $S$ is also a DFVS of $G'$ and $DFVS(G) \geq DFVS(G')$, since every added arc uses a vertex from $N^G(v^*)$. Otherwise, we can assume w.l.o.g.\ that $|N^G(v^*) \cap S| = 2$ and $v \in S$, since for every minimum DFVS of $G$, with $|N^G(v^*) \cap S| = 3$, there is a minimum DFVS $S^*$ with $N^G(v^*) \subseteq S^*$ of the same size. Furthermore, there cannot be a DFVS of $G$, which contains less than two elements from $N^G(v^*)$, since the elements in $N^G(v^*)$ form a path of four elements, where every arc is a bi-edge. So let $S$ be a minimum DFVS of $G$ such that $|N^G(v^*) \cap S| = 2$ and $v \in S$. We proceed by case distinction over the two neighbors of $v^*$ that are in $S$.

Case $N(v^*) \cap S = \{b,c\}$: In this case $S' = (S \setminus \{v^*\}) \cup \{a\}$ 
is a DFVS of $G'$ and $DFVS(G) \geq DFVS(G')$. Assume that on the contrary, there is a (w.l.o.g.) uncovered cycle. Then this cycle must use the vertex $d$, since the only added arcs, which do not have a vertex in $S'$ use $d$. Furthermore, the cycle must contain an arc between $N^G(a)$ and $d$ and between $N^G(d)$ and $d$. If only the latter holds true, then the same cycle is also present in $G-S$. If only the former holds true, then there is an equivalent cycle in $G-S$, where $d$ is replaced by $a$. Since the cycle is uncovered, the arcs must go into opposite direction, i.e., be of the form $(v_a,d)$ and $(d, v_d)$ or $(v_d,d)$ and $(d, v_a)$. Assume the first form, then we can transform the uncovered cycle into an uncovered path from $d$ to $a$ in $G$ by going from $v_a$ to $a$ instead of $d$. This is a contradiction to the assumption that there are no uncovered paths between $a$ and $d$. The argument for the latter form is analogous.

Case $N(v^*) \cap S = \{a,c\}$: In this case $S' = (S \setminus \{v^*\}) \cup \{d\}$ is a DFVS of $G'$ and $DFVS(G) \geq DFVS(G')$. The argument showing that $S'$ is a DFVS of $G'$ is analogous to that of the previous case.

Case $N(v^*) \cap S = \{b,d\}$: In this case $S' = (S \setminus \{v^*\}) \cup \{a\}$ is a DFVS of $G'$ and $DFVS(G) \geq DFVS(G')$. The argument showing that $S'$ is a DFVS of $G'$ is analogous to that of the first case.

As for the other direction, let $S'$ be a minimum DFVS of $G'$. If $N^G(v^*) \subseteq S'$, then $S'$ is also a DFVS of $G$ and $DFVS(G) \leq DFVS(G')$. Otherwise, we know that $|N^G(v^*) \cap S'| = 3$, since $G'[N^G(v^*)]$ is a diclique. 

Case $N^G(v^*) \cap S' = \{a,b,c\}$: Then $S = (S' \setminus \{a\}) \cup \{v^*\}$ is a DFVS of $G$ and $DFVS(G) \leq DFVS(G')$. Assume that on the contrary there is a cycle. Then this cycle must contain $a$. However, since $d$ has the same neighbors as $a$ in $G'$ this implies the existence of an equivalent cycle with $a$ replaced by $d$ in $G'$, which is a contradiction to $d \not\in S'$.

Case $N^G(v^*) \cap S' = \{a,b,d\}$: Then $S = (S' \setminus \{a\}) \cup \{v^*\}$ is a DFVS of $G$ and $DFVS(G) \leq DFVS(G')$ by an analogous argument as above.

The other two cases follow from symmetric arguments.
\end{proof}
In order to prove soundness of UNCONFINED, we use the following definitions inspired by~\cite{xiao2013confining}.

For a set $A \subseteq V(G)$ such that $G[A]$ is acyclic let 
\[
N_c(A) = \SB u \in V(G) \SM G[A\cup \{u\}] \text{ is cyclic}\SE,
\]
i.e., the neighbors of any vertex in $A$ such that there is a cycle through $A$ and the vertex. A vertex $u \in N_c(A)$ is called a directed child of $A$ if it has exactly two arcs that are shared with vertices in $A$ (i.e., $|N_{succ}(u) \cap A| + |N_{pred}(u) \cap A| = 2$). The vertices in $A$ that share arcs with $u$ are called its \emph{parents}.
\begin{lemma}
\label{lem:child}
Let $A$ be a set of vertices from $G$ such that
 \begin{itemize}
    \item $G[A]$ is acyclic, and
    \item for every minimum DFVS $S$ of $G$ it holds that $A \cap S = \emptyset$.
\end{itemize}
Then for each directed child $u$ of $A$ no minimum DFVS of $G$ contains all vertices $w \in N(u) \setminus (N_c(A) \cup A)$.
\end{lemma}
\begin{proof}
Assume that there is a minimum DFVS $S$ of $G$ such that $S \cap (N(u) \setminus (N_c(A) \cup A)) = (N(u) \setminus (N_c(A) \cup A))$ for some directed child $u \in N_c(A)$. The parents $p_1, p_2$ of $u$ are in $A$ and thus by assumption not in $S$. Furthermore, since $S$ is a DFVS and $A \cap S = \emptyset$, we know that $N_c(A) \subseteq S$ since $N_c(A)$ contains the vertices $v$ such that $G[A\cup \{v\}]$ is cyclic.
It follows that $N(u) \setminus A \subseteq S$ and thus $N(u) \setminus A \subseteq S'$, where
$S' = (S \cup \{p_1\}) \setminus \{u\}$ (or $(S \cup \{p_2\}) \setminus \{u\}$). Recall that $u$ is a directed child of $A$ and therefore only has two arcs that go from or to $A$, which use the parents. Since we put one of $u$'s parents into $S'$ there is at most one arc that uses $u$ in the graph $G - S'$, implying that $u$ cannot be contained in any cycle in $G - S'$.
Thus, the removal of $u$ can be compensated by the addition of $p_1$ (or $p_2$) and we see that $S'$ is a minimum DFVS of $G$, which shares a vertex with $A$. This is a contradiction.
\end{proof}

\begin{theorem}
Let $G$ be a digraph. After applying UNCONFINED to vertex $v$ resulting in $G'$, it holds that for every minimum DFVS $S$ of $G'$ the set $S\cup \{v\}$ is a minimum DFVS of $G$.
\end{theorem}
\begin{proof}
The result follows from Lemma~\ref{lem:child}. The idea of Algorithm~\ref{alg:dunconfined} 
is the following: We start by assuming that $S = \{v\}$ has an empty intersection with any minimum DFVS. If this leads to a contradiction together with Lemma~\ref{lem:child} then $v$ must be contained in some DFVS and the theorem follows.

In particular, we get a contradiction, when there is a directed child $u$ such that $|N(u) \setminus (N_c(A) \cup S)| = 0$. In this case, we know that there is a DFVS of $G$ that contains $v$. If there is no immediate contradiction but there is a $u \in N_c(A)$ with $|N(u) \setminus (N_c(A) \cup S)| = 1$, then we know that the unique vertex $v' \in N(u) \setminus (N_c(A) \cup S)$ must also not be contained in any minimum DFVS of $G$, which means that we can extend $S$ by adding $v'$.  
\end{proof}
}

\bibliographystyle{plainnat}
\bibliography{main.bib}

\end{document}